\documentclass[12pt, a4paper, reqno]{amsart}
\usepackage{mathrsfs}
\usepackage{bbm}
\usepackage{amsmath,amscd,amssymb,latexsym}
\usepackage[all]{xy}
\usepackage{color}
\usepackage[section]{placeins}
\usepackage[latin1]{inputenc}

\usepackage{algorithm}
\usepackage{algpseudocode}
\usepackage{amsmath}
\usepackage{cases}
\usepackage{arydshln}
\usepackage{cite}



\input xypic

\evensidemargin=\oddsidemargin
\allowdisplaybreaks[4]

\DeclareMathVersion{can}
\DeclareMathAlphabet{\can}{OT1}{cmss}{m}{n}
\newtheorem{thm}{Theorem}[section]
\newtheorem{cor}[thm]{Corollary}

\newtheorem{lem}[thm]{Lemma}
\newtheorem{prop}[thm]{Proposition}
\newtheorem{rem}[thm]{Remark}
\newtheorem{exa}[thm]{Example}
\theoremstyle{definition}
\newtheorem{defn}[thm]{Definition}
\theoremstyle{fact}

\theoremstyle{conjecture}

\newtheorem{alg}[thm]{Algorithm}

\numberwithin{equation}{section}
\begin{document}
%\large

\title[Hulls of generalized Reed-Solomon codes]{Determining hulls of generalized Reed-Solomon codes from algebraic geometry codes}

\author [Jia] {Xue Jia}
\address{\rm School of Mathematics, Nanjing University of Aeronautics and Astronautics, Nanjing,  211100, P. R. China and State Key Laboratory of Cryptology, P.O. Box 5159, Beijing 100878, P. R. China}
\email{jiaxue0904@163.com}

\author[Yue]{Qin Yue}
\address{\rm School of Mathematics, Nanjing University of Aeronautics and Astronautics, Nanjing,  211100, P. R. China and State Key Laboratory of Cryptology, P.O. Box 5159, Beijing 100878, P. R. China}
\email{yueqin@nuaa.edu.cn}

\author [Sun] {Huan Sun}
\address{\rm School of Mathematics, Nanjing University of Aeronautics and Astronautics, Nanjing, 211100, P. R. China and State Key Laboratory of Cryptology, P.O. Box 5159, Beijing 100878, P. R. China}
\email{sunhuan6558@163.com}

\author [Sui] {Junzhen Sui}
\address{\rm School of Mathematics and Information Science, Henan Normal University, Xinxiang, 453007, China}
\email{suijunzhen@163.com}

\thanks{The paper is supported by the National Natural Science Foundation of China ( No. 62172219 and No.
 12171420), the Natural Science Foundation of Shandong Province under Grant ZR2021MA046, the Natural Science Foundation of Jiangsu Province under Grant BK20200268, the Innovation Program for Quantum Science and Technology under Grant 2021ZD0302902.}



%\date{\today}%
%\dedicatory{}%
%\commby{}%
% ----------------------------------------------------------------

\begin{abstract}
In this paper, we provide  conditions that hulls of  generalized Reed-Solomon (GRS)  codes are also GRS codes from algebraic geometry codes. If the conditions are not satisfied, we provide a method of linear algebra to find the bases of hulls of GRS codes and give formulas to compute their dimensions. Besides, we explain that the conditions are too good to be improved by some examples. Moreover, we show self-orthogonal and self-dual GRS codes.
\end{abstract}
\maketitle
%\tableofcontents
% ----------------------------------------------------------------
%\setcounter{section}{-1}
\section{Introduction}
Let $\Bbb F_q$ be a finite field with $q$ elements and $\Bbb F_q[z]$ be the polynomial ring over $\Bbb F_q$, where   $q$ is a power prime. An $[n,k,d]$ linear code $\mathcal C$ over $\Bbb F_q$ is a $k$-dimensional subspace of $\Bbb F_q^n$ with length $n$ and minimum distance $d$. It is well known that the parameters $n,k$, and $d$  satisfy the Singleton bound:   $d\leq n-k+1$. A linear code with $d=n-k+1$ is called a maximum distance separable (MDS) code, which is extensively used in communications and has good applications in minimum storage codes and quantum codes \cite{CL,CH1,CH2,FF,FFL,HR,J,JX,LLL,RH,SY,Y}. An important  class of MDS codes are generalized Reed-Solomon (GRS) codes, which can be encoded and decoded efficiently. Moreover, some of the most important codes for practical use (such as BCH codes and Goppa codes) can be represented as subcodes of GRS codes in a natural way.

For two vectors $\boldsymbol x= (x_1,\ldots,x_n)$ and $\boldsymbol y= (y_1,\ldots,y_n)$ of $\Bbb F_q^n$, the Euclidean inner product of $\boldsymbol x$ and $\boldsymbol y$ is defined as
$\langle\boldsymbol x,\boldsymbol y\rangle=\sum_{i=1}^n x_i y_i.$
Let $\mathcal C$ be a linear code over $\Bbb F_q$, then  the Euclidean dual code of $\mathcal C$ is defined by
$$
\mathcal C^\perp=\{\boldsymbol x \in \Bbb F_q^n \mid \langle\boldsymbol x,\boldsymbol y\rangle=0 ~~\text{for all}~~ \boldsymbol y \in \mathcal C\}.
$$
The Euclidean hull of $\mathcal C$ is defined as $Hull(\mathcal C)=\mathcal C \cap \mathcal C^\perp$, which was originally introduced in $1990$ by Assmus et al.\cite{AK} to classify finite projective planes. Extensive research has shown that  hulls of linear codes play an important role in coding theory. Specifically, the hull of a linear code has become a central issue for the complexity of algorithms for computing the automorphism group of a linear code and for checking permutation equivalence of two linear codes \cite{L,L1,SN,SN1,SS}. There has been a large number of studies on the hulls of linear codes, such as \cite{FFL,GY,CH}. Recently, researchers have focused their attention on Euclidean and Hermitian hulls of GRS codes due to the wide applications in quantum communications \cite{CLL,FFL,GY,WY1}.

It is very interesting to consider  some cases of hulls of linear codes. The linear code $\mathcal C$ satisfying $Hull(\mathcal C) = \{0\}$ is called a linear complementary dual (LCD) code, which has been intensively investigated in \cite{M,J,CL,SY}.
The linear code $\mathcal C$ satisfying $Hull(\mathcal C) = \mathcal C$ (resp. $\mathcal C^\perp$) is called a self-orthogonal (resp. dual-containing) code. In particular, the linear code $\mathcal C$ satisfying $\mathcal C=\mathcal C^\perp$ is called a self-dual code. As everyone knows, an $[n,k,d]$ code $\mathcal C$ is self-dual, equivalent to $\dim (Hull(\mathcal C))=\frac{n}{2}$ for an even integer $n$. Considerable research efforts have been devoted to the construction of MDS self-orthogonal and MDS self-dual codes \cite{CH1,CH2,RH,FF,JX,LLL,Y,WY,WY1}.

Lately, Gao et al. \cite{GY} gave a method to explore the Euclidean hulls of GRS codes in terms of Goppa codes. Specifically, if $GRS_k(\boldsymbol \alpha, \boldsymbol v)$ is a $k$-dimensional GRS code over $\Bbb F_q$ associated with $\boldsymbol \alpha$ and $\boldsymbol v$ and $\boldsymbol v=(v_1,\ldots,v_n)$
satisfies $v_i=g(\alpha_i)^{-1}$ for a polynomial $g(z)\in \Bbb F_q[z]$ of degree $k$, $i=1,\ldots,n$, then $GRS_k(\boldsymbol \alpha, \boldsymbol v)$ is called a GRS code via Goppa code. They determined the hulls of GRS codes via Goppa codes and showed the hulls are also GRS codes via Goppa codes. Most recently, by expressing a Reed-Solomon code $RS_k(\boldsymbol \alpha)$ as a rational algebraic geometry code, Chen et al. \cite{CLL} completely determined the dimension of the hull $RS_k(\boldsymbol \alpha)\cap RS_k(\boldsymbol \alpha)^\perp$ in terms of the degree of some polynomials. However, they did not decide the specific form of the hull. Based on \cite{CLL} and \cite{GY}, we shall provide  conditions that  hulls of  GRS codes are also  GRS codes  from algebraic geometry codes, which generalize the results in \cite{GY}. If the conditions are not satisfied, then we shall give a method of linear algebra to find the bases of hulls of GRS codes and give formulas to compute the dimensions of their hulls. Besides, we explain that the conditions are too good to be improved by some examples.
Moreover, we show self-orthogonal or self-dual GRS codes.

This paper is organized as follows. In section 2, we recall some definitions and important results about rational function fields. In section 3, we first represent a GRS code and its dual code as rational algebraic geometric codes, then we provide  conditions that hulls of  GRS codes are also GRS codes. If they are not satisfied, we give a method to find the bases of the hulls and give formulas to compute the dimensions of their hulls. Moreover, we show self-orthogonal and self-dual GRS codes. In section 4, we conclude the paper.

\section{Preliminaries}
\subsection{\small{Rational function field}}\

In this subsection, we shall recall some basic knowledge and important results of rational function fields (in detail, see \cite{S}).

\subsubsection{\small{Rational function field and place}}\

In this paper, we always assume that $F=\Bbb F_q(z)$ is the rational function field over $\Bbb F_q$,   where $z$ is transcendental over $\Bbb F_q$ and $\Bbb F_q$ is called the field of constants of $F/\Bbb F_q$. For an irreducible monic polynomial $p(z)\in \Bbb F_q[z]$, the corresponding valuation ring of $F/\Bbb F_q$ is defined by
\begin{equation} \label{ring p}
\mathcal O_{p(z)}:=\left\{\frac{f(z)}{g(z)} ~\Big|~ f(z), g(z) \in \Bbb F_q[z], p(z) \nmid g(z)\right\};
\end{equation}
and it has a unique maximal ideal $P_{p(z)}=\mathcal O_{p(z)}\setminus \mathcal O_{p(z)}^*$ defined by
\begin{equation} \label{ideal p}
P_{p(z)}=\left\{\frac{f(z)}{g(z)}~\Big|~ f(z), g(z) \in \Bbb F_q[z], p(z)| f(z), p(z) \nmid g(z)\right\},
\end{equation}
where $\mathcal O_{p(z)}^*$ is the group of units of $\mathcal O_{p(z)}$.

%When $p(z)$ is linear, i.e. $p(z)=z-\alpha$ with $\alpha \in \Bbb F_q$, we abbreviate and write $$P_\alpha:=P_{z-\alpha}.$$
There is another valuation ring of $F/\Bbb F_q$, namely
\begin{equation} \label{ring inf}
\mathcal O_\infty:=\left\{\frac{f(z)}{g(z)}~ \Big|~ f(z), g(z) \in \Bbb F_q[z], \deg f(z) \leq \deg g(z) \right\}
\end{equation}
with a unique maximal ideal
\begin{equation} \label{ideal inf}
P_\infty=\left\{\frac{f(z)}{g(z)}~\Big|~ f(z), g(z) \in \Bbb F_q[z], \deg f(z) < \deg g(z)\right\}.
\end{equation}
In (\ref{ideal p}), the maximal ideal $P_{p(z)}=(p(z))$ is a principal ideal and is called a place of $F/\Bbb F_q$, where $p(z)$ is called a prime element for $P_{p(z)}$; in (\ref{ideal inf}), the maximal ideal $P_\infty=(\frac 1x)$ is also a principal ideal and is called the infinite place of $F/\Bbb F_q$, where $\frac 1x$ is called a prime element for $P_\infty$. Denote
$$
\Bbb P_{F}:=\left\{P\mid P ~~\text{is a place of}~~ F/\Bbb F_q\right\}.
$$

In general, for a place $P\in \Bbb P_F$ and a prime element $t$ for $P$, it is well known that every $0\neq x \in F$ has a unique representation $x = t^n u$ with $u \in \mathcal O_P^*$ and $n\in \Bbb Z$. Then we define a discrete valuation of $F/\Bbb F_q$ as
$v_P:F \rightarrow \Bbb Z \cup \{\infty\}$ such that $v_P(x):=n$ and $v_P(0):=\infty$, thus $v_P(xy)=v_P(x)+v_P(y)$ and $v_P(u)=0$ for $u \in \mathcal O_P^*$. For a place $P\in \Bbb P_F$, we have
$\mathcal O_P =\left\{x \in F \mid v_P(x) \geq 0\right\}$,
$\mathcal O_P^* =\left\{x \in F \mid v_P(x)=0\right\}$, and
$P =\left\{x \in F \mid v_{P}(x)>0\right\}$.
For $x\in F$ and $P\in \Bbb P_F$, we say $P$ is a zero of $x$
of order $m$ if $v_P (x) = m > 0$; $P$ is a pole of $x$ of order $m$ if $v_P (x) = -m < 0$.

Let $P$ be a place of $F/\Bbb F_q$ and $\mathcal O_P$ be its valuation ring. Since $P$ is a
maximal ideal, the residue class ring $\mathcal O_P/P$ is a field. Define the residue class map with respect to $P$ as follows:
\begin{eqnarray*}
&F\longrightarrow\mathcal O_P/P \cup \{\infty\}& \\
&x\in \mathcal O_P\longmapsto x(P),& \\
&x\in F\setminus\mathcal O_P\longmapsto\infty,&
\end{eqnarray*}
where $x(P):=x+P$ is the residue class of $x$ modulo $P$.
It is clear that $\mathcal O_P\supseteq \Bbb F_q$ and $\deg P:=[\mathcal O_P/P: \Bbb F_q]$ is called the degree of $P$. A place of degree one of $F/\Bbb F_q$ is called a rational place of $F/\Bbb F_q$. In the rational function field $F$,  $\deg P_{p(z)}=\deg p(z)$, $\deg P_{\infty}=1$, and all rational places are $P_{z-\alpha}$ and $P_\infty$, where $\alpha \in \Bbb F_q$.
\subsubsection{\small{Divisor, the Riemann-Roch space, the Riemann-Roch Theorem}}\

The divisor group $Div(F)$ of $F/\Bbb F_q$ is defined by
$$
Div(F)=\{D \mid D=\sum_{P\in \Bbb P_F}n_P P ~~ \text{with}~~ n_P \in \Bbb Z, ~~\text{almost all} ~~ n_P=0\}.
$$
The zero element of $Div(F)$ is the divisor
$0:=\sum_{P\in \Bbb P_F} r_P P, ~~\text{all}~~ r_P=0.$
The support of a divisor $D$ is defined as $supp(D):=\{P\in \Bbb P_F \mid n_P \neq 0\}$.
For a place $Q\in \Bbb P_F$ and a divisor $D=\sum n_P P\in Div(F)$,  define $v_Q(D):=n_Q$ and
$$D=\sum_{P\in \Bbb P_F}v_P(D)P.$$
A partial ordering on $Div(F)$ is defined by $$D_1\leq D_2 \Longleftrightarrow v_P(D_1)\leq v_P(D_2)\mbox{ for all }P\in \Bbb P_F.$$ A divisor $D\geq 0$ is called positive (or effective). The degree of a divisor $D$ is defined as $\deg D:=\sum_{P\in \Bbb P_F}v_P(D)\deg P$.

By\cite [Corollary 1.3.4]{S}, a nonzero element $x\in F$ has only finitely many zeros and poles in $\Bbb P_F$. Let $0\neq x\in F$ and denote by $Z$ (resp. $N$) the set of zeros (resp. poles) of $x$ in $\Bbb P_F$. Then we define
$$
\begin{aligned}
&(x)_{0}:=\sum_{P \in Z} v_{P}(x) P , ~~\text{the zero divisor of}~~  x, \\
&(x)_{\infty}:=\sum_{P \in N}\left(-v_{P}(x)\right) P , ~~\text{the pole divisor of}~~  x , \\
&(x):=(x)_{0}-(x)_{\infty}=\sum_{P \in \Bbb P_F}v_{P}(x) P , ~~\text{the principal divisor of}~~  x .
\end{aligned}
$$
It is clear that $(x)_{0}\geq 0, (x)_{\infty}\geq 0$. By $v_P(xy)=v_P(x)+v_P(y)$,  $(xy)=(x)+(y)$. Since the elements in $\Bbb F_q$ have no zeros and no poles, there is a characterization of the elements in $\Bbb F_q$:
$$x\in \Bbb F_q \Longleftrightarrow (x)=0.$$
 The set of divisors
$$
Princ(F):=\{(x)\mid 0\ne x\in F\}
$$
is called the group of principal divisors of $F/\Bbb F_q$, which is a subgroup of $Div(F)$.

For two divisors $D_1=\sum_{P\in \Bbb P_F}v_P(D_1)P$ and $D_2=\sum_{P\in \Bbb P_F}v_P(D_2)P$,
we define the following operations:
\begin{eqnarray*}
&&D_1+D_2=\sum_{P\in \Bbb P_F}(v_P(D_1)+v_P(D_2))P; \\
&&D_1\cap D_2=\sum_{P\in \Bbb P_F} \min\{v_P(D_1),v_P(D_2)\}P; \\
&&D_1\cup D_2=\sum_{P\in \Bbb P_F} \max\{v_P(D_1),v_P(D_2)\}P.
\end{eqnarray*}

For a divisor $A\in Div(F)$, define the Riemann-Roch space associated with $A$ as
$$
\mathscr L(A):=\{x\in F \mid (x)\geq -A\}\cup \{0\}=\{x\in F \mid v_P(x)\geq -v_P(A) ~~\text{for all}~~ P\in \Bbb P_F\}.
$$
It is well known that $\mathscr L(A)$ is a vector space over $\Bbb F_q$ and the integer $\ell(A):=\dim \mathscr L(A)$ is called the dimension of the divisor $A$.
In terms of Riemann-Roch space, we list the following properties in \cite{S} for later use:

(1) If $A < 0$, then $\mathscr L(A)=\{0\}$ and $\ell(A)=0$.

(2) All principal divisors have degree zero. More precisely: let $x\in F\setminus \Bbb F_q$ and $(x)_0$ resp. $(x)_\infty$ denote the zero resp. pole divisor of $x$. Then
$$
\deg (x)_0 = \deg (x)_\infty.
$$

(3) If $A'\in Div(F)$ is a divisor satisfying $A'=A+(x)$ for some $x \in F$, then $\mathscr L(A)\cong \mathscr L(A')$, $\ell(A)=\ell(A')$, and $\deg A=\deg A'$.

(4) Let $D_1$ and $D_2$ be two divisors of $F/\Bbb F_q$. Then
$$\mathscr L(D_1)\cap \mathscr L(D_2)=\mathscr L(D_1\cap D_2).$$

Now there is a very important fact to help us understand the above concepts.

For $0 \neq x \in F$,  $x=a\cdot \frac{f(z)}{g(z)}$ with $a \in \Bbb F_q\setminus \{0\}$, where
$f(z),g(z)\in \Bbb F_q[z]$ are  monic and relatively prime, i.e., $\gcd(f(z),g(z))=1$. Let
$$
f(z)=\prod_{i=1}^r p_i(z)^{n_i},g(z)=\prod_{j=1}^s q_j(z)^{m_j}
$$
with pairwise distinct irreducible monic polynomials $p_i(z),q_j(z)\in \Bbb F_q[z]$ and positive integers $n_i,m_j,i=1,\ldots,r,j=1,\ldots,s$.
Then the principal divisor of $x$ in $Div(F)$ appears thus :
$$
(x)=\sum_{i=1}^r n_i P_i-\sum_{j=1}^s m_j Q_j+(\deg g(z)-\deg f(z))P_\infty,
$$
where $P_i$ and $Q_j$ are the places corresponding to $p_i(z)$ and $q_j(z)$, respectively.

One of the most important problems in the theory of algebraic function fields is to calculate the dimension of a divisor. The answer to the problem will be given by the Riemann-Roch Theorem. Since the description of the Riemann-Roch Theorem involves many concepts, we shall introduce the theorem briefly here. Readers interested in this can refer to section 1.5 in \cite{S}.

The genus $g$ of $F/\Bbb F_q$ is defined as
$
g:= \max\{\deg A-\ell(A)+1\mid A\in Div(F)\}.
$
\cite [Example 1.4.18]{S} tells us the rational function field $F$ has genus $g=0$, which is to say the genus $g$ in this paper is always equal to $0$.

\begin{thm}(Riemann-Roch Theorem) \label{RR thm}  (\cite[Theorem  1.5.15]{S})
Let $W$ be a canonical divisor of $F/\Bbb F_q$. Then for each divisor $A\in Div(F)$,
$$
\ell(A)=\deg A+1-g+\ell(W-A).
$$
\end{thm}
In fact, for a canonical divisor $W$, $\deg W=2g-2=-2$ by  \cite [Corollary 1.5.16] {S} and $g=0$. If $A$ is a divisor of $F/\Bbb F_q$ with $\deg A> 2g-2=-2$, then
$\ell(A)=\deg A+1-g=\deg A+1$.
\subsection{\small{Rational algebraic geometry code}}\

In this subsection, we shall recall definitions and properties of rational algebraic  geometric (AG) codes and generalized Reed-Solomon codes (in detail, see \cite{S}).

Let $F=\Bbb F_q(z)$ be the rational function field (genus $g=0$), $P_1,\ldots,P_n$ pairwise distinct places of $F/\Bbb F_q$ of degree one, $D=P_1+\cdots+P_n$, and $G$ a divisor of $F/\Bbb F_q$ such that $supp(G)\cap supp(D)=\emptyset$. Consider the evaluation map
$$
ev_D: \mathscr L(G) \rightarrow \Bbb F_q^n, x \mapsto (x(P_1), \ldots, x(P_n)),
$$
which is a $\Bbb F_q$-linear map with $\ker (ev_D)=\mathscr L(G-D)$.
The rational algebraic geometry code $C_{\mathscr L}(D,G)$ associated with the divisors $D$ and $G$ of the rational function field $F$ is defined by
$$
C_{\mathscr L}(D,G):=\{(x(P_1),\ldots,x(P_n))\mid x\in \mathscr L(G)\}=ev_D(\mathscr L(G)) \subseteq \Bbb F_q^n.
$$
Moreover, $C_{\mathscr L}(D,G)$ is a code with parameters $[n,k,d]$, where
$$
k=\ell(G)-\ell(G-D) \quad\text{and}\quad d\geq n-\deg G.
$$
And if $\deg G<n$ and $\{x_1,\ldots,x_k\}$ is a basis of $\mathscr L(G)$, then
$\dim (\ker (ev_D))=0$, the evaluation $ev_D: \mathscr L(G) \rightarrow C_{\mathscr L}(D,G)$ is injective, and $C_{\mathscr L}(D,G)$ has a generator matrix:
\begin{equation}\label{AGgm}
M=\left(\begin{array}{cccc}
x_1(P_1) & x_1(P_2) & \cdots & x_1(P_n) \\
x_2(P_1) & x_2(P_2) & \cdots & x_2(P_n) \\
\vdots & \vdots & & \vdots \\
x_k(P_1) & x_k(P_2) & \cdots & x_k(P_n)
\end{array}\right).
\end{equation}
\begin{prop} (\cite[Proposition 2.3.2]{S}) \label{prop para}
Let $C=C_{\mathscr L}(D,G)$ be a rational AG code over $\Bbb F_q$ with parameters $n$ and $k$. Then\\
(1) $k=0$ if and only if $\deg G < 0$, and $k=n$ if and only if $\deg G > n-2$.\\
(2) For $0 \leq \deg G \leq n-2$, $k=1+\deg G$.\\
(3) $C^\perp$ is also a rational AG code.
\end{prop}

\begin{prop}(\cite[Proposition 2.2.10]{S})\label{AG dual}
Let $C_{\mathscr L}(D,G)$ be a rational AG code. Then
$$
C_{\mathscr L}(D,G)^\perp=C_{\mathscr L}(D,D-G+(\eta)),
$$
where $\eta$ is a Weil differential (see \cite[Definition 1.5.6]{S}) such that $v_{P_i}(\eta)=-1$ and $\eta_{P_i}(1)=1$ for $i=1,\ldots,n$.
\end{prop}
\begin{lem}(\cite[Lemma 2.3.6]{S})\label{wd}
Consider the rational function field $F=\Bbb F_q(z)$ and $n$ distinct elements $\alpha_1,\ldots, \alpha_n \in \Bbb F_q$. Let $P_i:=P_{z-\alpha_i}$ be the rational places of $z-\alpha_i$ and $h(z)=\prod_{i=1}^n(z-\alpha_i)$. Then there exists a Weil differential $\eta$ of $F/\Bbb F_q$ with
$$
(\eta)=(h'(z))-(h(z))-2P_{\infty}
$$
satisfying
$$
v_{P_i}(\eta)=-1 \quad\text{and} \quad \eta_{P_i}(1)=1, i=1,\ldots,n,
$$
where $h'(z) \in \Bbb F_q[z]$ is the derivative of the polynomial $h(z)$.
\end{lem}
\begin{rem}
In this paper, we always assume that $\gcd (n,q)=1$, so that $\deg h'(z)=n-1$  to omit unnecessary discussion.
\end{rem}
 There is an important result in the following   research.
\begin{lem}(\cite[Lemma 7]{CLL}) \label{chen}
Let $G_1$ and $G_2$ be two divisors of the rational function field $F$ satisfying $\deg(G_1\cap G_2)> -2$.
Let $D=P_1+\cdots+P_n$ be a divisor of $F/\Bbb F_q$, where $P_1,\ldots,P_n$ are pairwise distinct places of $F/\Bbb F_q$ of degree one. Then \\
(1) $\mathscr L(G_1)+\mathscr L(G_2)=\mathscr L(G_1\cup G_2)$. \\
(2) If $supp(G_i)\cap supp(D)=\emptyset, ~~ i=1,2$, then
$$
C_{\mathscr L}(D,G_1)+C_{\mathscr L}(D,G_2)=C_{\mathscr L}(D,G_1\cup G_2).
$$
\end{lem}
Now we recall the definition of a generalized Reed-Solomon code, which is closely related to a rational algebraic geometric code. %and represent generalized Reed-Solomon code and its dual code as rational algebraic geometric codes.
\begin{defn}\label{GRS}
Let $n$ be a positive integer with $1<n<q$. Let $\boldsymbol \alpha=(\alpha_1,\ldots, \alpha_n)$, where $\alpha_i$ are distinct elements of $\Bbb F_q$, and $\boldsymbol v=(v_1,\ldots, v_n)$, where $v_i$ are nonzero (not necessarily distinct) elements of $\Bbb F_q$, $i=1,\ldots,n$. For $1\leq k< n$, the generalized Reed-Solomon (GRS for short) code is defined by
\begin{equation} \label{GRS eq}
GRS_k(\boldsymbol \alpha, \boldsymbol v)=\{(v_1f(\alpha_1),\ldots,v_nf(\alpha_n)):f(z) \in \Bbb F_q[z], \deg f(z)\leq k-1\}.
\end{equation}

If $v_i=1$ for $i=1,\ldots,n$, then $GRS_k(\boldsymbol \alpha, \boldsymbol v)$ is called a Reed-Solomon (RS) code. It is clear that $GRS_k(\boldsymbol \alpha, \boldsymbol v)$ has a generator matrix as follows:
\begin{equation}\label{GRS gm}
\left(\begin{array}{cccc}
v_1 & v_2 & \cdots & v_n \\
v_1\alpha_1 & v_2\alpha_2 & \cdots & v_n\alpha_n \\
\vdots & \vdots & & \vdots \\
v_1\alpha_1^{k-1} & v_2\alpha_2^{k-1} & \cdots & v_n\alpha_n^{k-1}
\end{array}\right).
\end{equation}
It is well known that $GRS_k(\boldsymbol \alpha, \boldsymbol v)$ is an $[n,k,n-k+1]$ MDS code. In fact,  all rational AG codes over $\Bbb F_q$ of length $n\leq q$ are GRS codes \cite [Proposition 2.3.3 (a)]{S} .
\end{defn}

\begin{lem} \label{dual}
$$
GRS_k(\boldsymbol \alpha, \boldsymbol v)\cap GRS_k(\boldsymbol \alpha, \boldsymbol v)^\perp=(GRS_k(\boldsymbol \alpha, \boldsymbol v)+GRS_k(\boldsymbol \alpha, \boldsymbol v)^\perp)^\perp.
$$
\end{lem}
\begin{proof}
It is clear that $GRS_k(\boldsymbol \alpha, \boldsymbol v)\subseteq GRS_k(\boldsymbol \alpha, \boldsymbol v)+GRS_k(\boldsymbol \alpha, \boldsymbol v)^\perp$, then
$$
(GRS_k(\boldsymbol \alpha, \boldsymbol v)+GRS_k(\boldsymbol \alpha, \boldsymbol v)^\perp)^\perp\subseteq GRS_k(\boldsymbol \alpha, \boldsymbol v)^\perp.
$$
Similarly,
$$
(GRS_k(\boldsymbol \alpha, \boldsymbol v)+GRS_k(\boldsymbol \alpha, \boldsymbol v)^\perp)^\perp\subseteq GRS_k(\boldsymbol \alpha, \boldsymbol v).
$$
Then
$$(GRS_k(\boldsymbol \alpha, \boldsymbol v)+GRS_k(\boldsymbol \alpha, \boldsymbol v)^\perp)^\perp\subseteq GRS_k(\boldsymbol \alpha, \boldsymbol v)\cap GRS_k(\boldsymbol \alpha, \boldsymbol v)^\perp.$$

On the other hand,
$$\begin{aligned}
&\dim (GRS_k(\boldsymbol \alpha, \boldsymbol v)+GRS_k(\boldsymbol \alpha, \boldsymbol v)^\perp)^\perp
=n-\dim (GRS_k(\boldsymbol \alpha, \boldsymbol v)+GRS_k(\boldsymbol \alpha, \boldsymbol v)^\perp)\\
&=\dim(GRS_k(\boldsymbol \alpha, \boldsymbol v))+
\dim(GRS_k(\boldsymbol \alpha, \boldsymbol v)^\perp)-\dim (GRS_k(\boldsymbol \alpha, \boldsymbol v)+GRS_k(\boldsymbol \alpha, \boldsymbol v)^\perp)\\
&=\dim(GRS_k(\boldsymbol \alpha, \boldsymbol v)\cap GRS_k(\boldsymbol \alpha, \boldsymbol v)^\perp).\end{aligned}$$
Therefore, we obtain the desired result.
\end{proof}

Chen et al.\cite{CLL} expressed a $k$-dimensional Reed-Solomon (RS) code $RS_k(\boldsymbol \alpha)$ over $\Bbb F_q$ associated with $\boldsymbol \alpha=(\alpha_1,\ldots,\alpha_n)$ as a rational algebraic geometry code, and they completely determined the dimension of the hull $RS_k(\boldsymbol \alpha)\cap RS_k(\boldsymbol \alpha)^\perp$ in terms of the degree of the derivative of $h(z)$ and some relevant polynomials.
 Gao et al.\cite{GY} gave  a necessary and sufficient condition that the dual code of a GRS code via Goppa code is also a GRS code via Goppa code. Under the above condition, they showed that hulls of GRS codes via Goppa codes are also GRS codes via Goppa codes.
Naturally, we shall consider the following questions.

{\bf Questions:}

(1) What is the hull of a GRS code denoted by a rational algebraic geometry code?

(2) What is the condition that the hull of a GRS code denoted by a rational algebraic geometry code is also a GRS code?

In this paper, we shall provide conditions that hulls of GRS codes denoted by rational algebraic geometry codes are also GRS codes. If they are not satisfied, we shall give a method of linear algebra to find the bases of hulls of GRS codes.
\section{Hulls of generalized Reed-Solomon Codes}

In this section, we first represent a GRS code and its dual code as rational algebraic geometry codes. After that, we shall provide conditions that hulls of GRS codes are also GRS codes from algebraic geometric codes. If the conditions are not satisfied, then we shall give a method of linear algebra to find the bases of hulls of GRS codes and give formulas to compute  their dimensions. Besides, we explain that the conditions are too good to be improved by some examples.

%and explain that they are too good to be improved by some examples. Moreover, we obtain some conditions to construct self-orthogonal GRS codes and self-dual GRS codes.

In this paper, we always assume the following notations:
\begin{center}
\begin{tabular}{ll}

  $\Bbb F_q^*=\Bbb F_q\backslash \{0\}$ consists of all  non-zero elements of $\Bbb F_q$.\\
  $F=\Bbb F_q(z)$ is the rational function field over $\Bbb F_q$. \\
  $\boldsymbol v=(v_1,\ldots, v_n)$, where $v_i$, $i=1,\ldots,n$, are nonzero elements of $\Bbb F_q$. \\
  $\boldsymbol \alpha=(\alpha_1,\ldots, \alpha_n)$, where $\alpha_i$, $i=1,\ldots,n$, are distinct elements of $\Bbb F_q$. \\
  $h(z)=\prod_{i=1}^n(z-\alpha_i)\in \Bbb F_q[z]$. \\
  $h'(z) \in \Bbb F_q[z]$ is the derivative of the polynomial $h(z)$.\\
  $P_i:=P_{z-\alpha_i}$are the rational places of the irreducible polynomials $z-\alpha_i$,  $i=1,\ldots,n$.  \\
  $P_{\infty}$ is the infinite place of $F/\Bbb F_q$ defined by (\ref{ideal inf}). \\
  $D=P_1+\cdots+P_n=(h(z))+nP_{\infty} \in Div(F)$. \\
 \end{tabular}
\end{center}
\begin{prop}\label{GRS AG}
 Let a GRS code $GRS_k(\boldsymbol \alpha, \boldsymbol v)$ be defined by Definition \ref{GRS}. Then the $k$-dimensional GRS code $GRS_k(\boldsymbol \alpha, \boldsymbol v)$ can be represented by a rational algebraic geometry code as follows:
$$
GRS_k(\boldsymbol \alpha, \boldsymbol v)=C_{\mathscr L}(D,(k-1)P_{\infty}+(s(z))),
$$
where $s(z)\in \Bbb F_q[z]$ is a polynomial with $0\leq\deg s(z)< n$ and $s(\alpha_i)=\frac 1{v_i} \in \Bbb F_q^*$, $i=1,\ldots,n$.
\end{prop}
\begin{proof}
Note that the $k$-dimensional GRS code $GRS_k(\boldsymbol \alpha, \boldsymbol v)$ has a generator matrix defined by (\ref{GRS gm}). According to Lagrange interpolation, there exists a unique polynomial $s(z) \in \Bbb F_q[z]$ with $0\leq \deg s(z) <n$ such that
$$
s(\alpha_i)=\frac 1{v_i} \in \Bbb F_q^*, i=1,2,\ldots,n.
$$
Then it is easy to verify that $\{\frac 1{s(z)}, \frac z{s(z)},\ldots, \frac{z^{k-1}}{s(z)}\}$ is a basis of $\mathscr L((k-1)P_{\infty}+(s(z)))$.

In fact, by Theorem \ref{RR thm} and $\deg((k-1)P_{\infty}+(s(z)))=k-1 > -2$, $\dim(\mathscr L((k-1)P_{\infty}+(s(z))))=\deg((k-1)P_{\infty}+(s(z)))+1=k$.
On the other hand, there are $k$ $\Bbb F_q$-linearly independent elements: $\frac 1{s(z)}, \frac z{s(z)},\ldots, \frac{z^{k-1}}{s(z)}$, in $\mathscr L((k-1)P_{\infty}+(s(z)))$. Then $\{\frac 1{s(z)}, \frac z{s(z)},\ldots, \frac{z^{k-1}}{s(z)}\}$ is a basis of $\mathscr L((k-1)P_{\infty}+(s(z)))$.

  Hence
$C_{\mathscr L}(D,(k-1)P_{\infty}+(s(z)))$ has the  generator matrix as (\ref{GRS gm}),  i.e.,
$$
GRS_k(\boldsymbol \alpha, \boldsymbol v)=C_{\mathscr L}(D,(k-1)P_{\infty}+(s(z))).
$$
\end{proof}
\begin{rem}
Recall the $k$-dimensional GRS code via Goppa code $GRS_k(\boldsymbol \alpha, \boldsymbol v)$ in \cite{GY}, where $\boldsymbol v=(v_1,\ldots,v_n)$ satisfies $v_i=g(\alpha_i)^{-1}$ for a polynomial $g(z)\in\Bbb F_q[z]$ of degree $k$, $i=1,\ldots,n$.
What is different is that the polynomial $g(z)\in \Bbb F_q[z]$ in \cite{GY} is of degree $k$, while the polynomial $s(z)\in \Bbb F_q[z]$ in Proposition \ref{GRS AG} is of degree less than $n$, which is not necessarily $k$.
\end{rem}
\begin{prop}\label{GRS dual}
Let $GRS_k(\boldsymbol \alpha, \boldsymbol v)$ be the GRS code in Proposition \ref{GRS AG}. Then
$$
GRS_k(\boldsymbol \alpha, \boldsymbol v)^\perp=C_{\mathscr L}(D,(k-1)P_{\infty}+(s(z)))^\perp=C_{\mathscr L}(D,(n-k-1)P_{\infty}+(t(z)))
$$
if and only if
\begin{equation} \label{im eq}
s(z)t(z)=u(z)h(z)+ h'(z),
\end{equation}
where $u(z)\in \Bbb F_q[z]$ is of  degree at most $n-2$ and  $t(z)\in \Bbb F_q[z]$ is a polynomial with $0\leq\deg t(z)< n$ and $t(\alpha_i)=\frac { h'(\alpha_i)}{s(\alpha_i)}$, $i=1,2,\ldots,n$.
\end{prop}
\begin{proof}
By Proposition \ref{AG dual},  Lemma \ref{wd},  and the principal divisor $(h(z))=D-nP_{\infty}$,
$$
\begin{aligned}
&GRS_k(\boldsymbol \alpha, \boldsymbol v)^\perp=C_{\mathscr L}(D,(k-1)P_{\infty}+(s(z)))^\perp\\
&=C_{\mathscr L}(D,D-((k-1)P_{\infty}+(s(z)))+(h'(z))-(h(z))-2P_{\infty})\\
&=C_{\mathscr L}(D,-(s(z))+(h'(z))+(n-k-1)P_{\infty}).
\end{aligned}
$$
By the proof of  Proposition \ref{GRS AG}, it is easy to check that  $\{\frac {s(z)}{h'(z)},
\frac {s(z)}{h'(z)}z,\ldots, \frac {s(z)}{h'(z)}z^{n-k-1}\}$ is a basis of $\mathscr L(-(s(z))+(h'(z))+(n-k-1)P_{\infty})$, then $GRS_k(\boldsymbol \alpha, \boldsymbol v)^\perp$ has a generator matrix:
$$
H=\left(\begin{array}{ccc}
\frac{s(\alpha_1)}{h'(\alpha_1)} &  \cdots & \frac{s(\alpha_n)}{h'(\alpha_n)} \\
\frac{s(\alpha_1)}{h'(\alpha_1)}\alpha_1 &  \cdots & \frac{s(\alpha_n)}{h'(\alpha_n)}\alpha_n \\
\vdots &  & \vdots \\
\frac{s(\alpha_1)}{h'(\alpha_1)}\alpha_1^{n-k-1} &  \cdots & \frac{s(\alpha_n)}{h'(\alpha_n)}\alpha_n^{n-k-1}
\end{array}\right).
$$

According to Lagrange interpolation again, there exists a polynomial $t(z) \in \Bbb F_q[z]$ with $0\leq \deg t(z) <n$ such that
$$
t(\alpha_i)=\frac {h'(\alpha_i)}{s(\alpha_i)}, s(\alpha_i)t(\alpha_i)= h'(\alpha_i), i=1,2,\ldots,n.
$$

Hence  $h(z)\mid(s(z)t(z)- h'(z))$, i.e.,
$$
s(z)t(z)=u(z)h(z)+h'(z),
$$
where  $u(z)\in \Bbb F_q[z]$ is of  degree at most $n-2$ (by convention, we set $\deg 0=-\infty$).
Moreover,
$$
GRS_k(\boldsymbol \alpha, \boldsymbol v)^\perp=C_{\mathscr L}(D,(n-k-1)P_{\infty}+(t(z))),
$$
 which completes the proof.
\end{proof}

\begin{rem}
If $s(z)\in \Bbb F_q[z]$ is a polynomial of degree $k$ such that $s(\alpha_i)=\frac 1{v_i}\in \Bbb F_q^*$, $i=1,2,\ldots,n$, then $(k-1)P_{\infty}+(s(z))=(s(z))_0-P_{\infty}$ and $-(s(z))+(h'(z))+(n-k-1)P_{\infty}=(h'(z))+(n-1)P_{\infty}-(s(z))_0$.
From \cite [Proposition 2.3.11]{S}, we know that $C_{\mathscr L}(D,-(s(z))+(h'(z))+(n-k-1)P_{\infty})$ is just the Goppa code with the Goppa polynomial $s(z)$, while
$C_{\mathscr L}(D,(k-1)P_{\infty}+(s(z)))$ is the dual code of Goppa code, i.e., the GRS code $GRS_k(\boldsymbol \alpha, \boldsymbol v)$ via Goppa code in \cite{GY}.
\end{rem}

In the following, we always assume that $GRS_k(\boldsymbol \alpha, \boldsymbol v)=C_{\mathscr L}(D,(k-1)P_{\infty}+(s(z)))$ and $GRS_k(\boldsymbol \alpha, \boldsymbol v)^\perp=C_{\mathscr L}(D,(n-k-1)P_{\infty}+(t(z)))$ be the two GRS codes in Propositions \ref{GRS AG} and  \ref{GRS dual}, respectively, where   $$s(z)t(z)=u(z)h(z)+h'(z)$$ and $\deg u(z)\le n-2$.
Define
\begin{equation}\label{v}
\varepsilon=\left\{\begin{array}{ll}
-1, &\mbox{if }u(z)=0,\\ \deg u(z), &\mbox{if }u(z)\ne0.\end{array}\right.
\end{equation}
Then $\deg s(z)+\deg t(z)=n+\varepsilon$.

According to Lemmas \ref{chen} and \ref{dual}, we can transform the problem of studying the hull of a GRS code into the problem of studying the sum of a GRS code and its dual code.

\begin{thm} \label{main th} If
$k+\varepsilon -\deg d(z)-1 <\deg s(z)<k+\frac {\varepsilon}2$, then
\begin{equation} \label{th2 eq2}
Hull(GRS_k(\boldsymbol \alpha, \boldsymbol v))=C_{\mathscr L}(D, (d(z))+(\deg d(z)+\deg s(z)-k-\varepsilon-1)P_{\infty})
\end{equation} is a GRS code.

If $k+\frac {\varepsilon}2\leq\deg s(z)<k+\deg d(z)+1$, then
\begin{equation} \label{th2 eq1}
Hull(GRS_k(\boldsymbol \alpha, \boldsymbol v))=C_{\mathscr L}(D,(d(z))+(k+\deg d(z)-\deg s(z)-1)P_{\infty})
\end{equation} is a GRS code.

Moreover, if $k+\varepsilon -\deg d(z)-1<\deg s(z)<k+\deg d(z)+1$, then the hull of $GRS_k(\boldsymbol \alpha, \boldsymbol v)$ is also a GRS code.
\end{thm}
\begin{proof}
Denote divisors
\begin{equation}\label{case 2 G_1}
G_1=(k-1)P_{\infty}+(s(z))=(s(z))_0+(k-1-\deg s(z))P_{\infty}
\end{equation}
and
$$
G_2=(n-k-1)P_{\infty}+(t(z))=(t(z))_0+(n-k-1-\deg t(z))P_{\infty}.
$$
In fact,  $\deg s(z)+\deg t(z)=n+\varepsilon$ and
\begin{equation}\label{case 2 G_2}
G_2=(t(z))_0+(\deg s(z)-\varepsilon-k-1)P_{\infty}.
\end{equation}

In the following, we consider the condition $\deg(G_1\cap G_2)>-2$ in Lemma \ref{chen} in  two cases.

{\bf Case 1: $\deg s(z)<k+ \frac {\varepsilon}2$}. It is clear that
$$
\begin{aligned}
G_1 \cap G_2&=(d(z))_0+\min\{k-1-\deg s(z),\deg s(z)-\varepsilon-k-1\}P_{\infty}.
\end{aligned}
$$
Suppose that $k-1-\deg s(z)\le\deg s(z)-\varepsilon-k-1$. Then $\deg s(z)\ge k+\frac {\varepsilon}2$, which is contradictory. Hence $G_1\cap G_2=(d(z))_0+(\deg s(z)-\varepsilon-k-1)P_{\infty}$
and $\deg(G_1 \cap G_2)=\deg d(z)+\deg s(z)-\varepsilon-k-1$. Then $\deg(G_1\cap G_2)>-2$ if and only if $k+\varepsilon -\deg d(z)-1 < \deg s(z)<k+\frac {\varepsilon}2$.

If  $k+\varepsilon -\deg d(z)-1 <\deg s(z)<k+\frac {\varepsilon}2$, then by $\deg l(z)+\deg d(z)=\deg s(z)+\deg t(z)=n+\varepsilon$,

$$
\begin{aligned}
G_1 \cup G_2&=(l(z))_0+\max\{k-1-\deg s(z),\deg s(z)-\varepsilon-k-1\}P_{\infty}\\
&=(l(z))_0+(k-1-\deg s(z))P_{\infty}\\
&=(l(z))+(n+\varepsilon-\deg d(z)+k-1-\deg s(z))P_{\infty}.
\end{aligned}
$$
Hence by Lemmas \ref{chen} and \ref{dual},
$$
\begin{aligned}
Hull(GRS_k(\boldsymbol \alpha, \boldsymbol v))&=GRS_k(\boldsymbol \alpha, \boldsymbol v)\cap GRS_k(\boldsymbol \alpha, \boldsymbol v)^\perp=(GRS_k(\boldsymbol \alpha, \boldsymbol v)+GRS_k(\boldsymbol \alpha, \boldsymbol v)^\perp)^\perp \\
&=(C_{\mathscr L}(D,G_1)+C_{\mathscr L}(D,G_2))^\perp = C_{\mathscr L}(D,G_1\cup G_2)^\perp\\
&=C_{\mathscr L}(D,(l(z))+(n+\varepsilon+k-\deg d(z)-\deg s(z)-1)P_{\infty})^\perp.
\end{aligned}
$$
By $l(z)d(z)=s(z)t(z)=u(z)h(z)+h'(z)$ and Proposition \ref{GRS dual},
$$
Hull(GRS_k(\boldsymbol \alpha, \boldsymbol v))=C_{\mathscr L}(D,(d(z))+(\deg d(z)+\deg s(z)-k-\varepsilon-1)P_{\infty})
$$
is a GRS code.

{\bf Case 2: $\deg s(z)\ge k+\frac {\varepsilon}2$}. It is clear that
$$
\begin{aligned}
\deg(G_1 \cap G_2)&=\deg((d(z))_0+(k-1-\deg s(z))P_{\infty}) \\
&=\deg d(z)+k-1-\deg s(z).
\end{aligned}
$$
Then $\deg(G_1\cap G_2)>-2$ if and only if $k+\frac {\varepsilon}2\leq \deg s(z)<k+\deg d(z)+1$.

If $k+\frac {\varepsilon}2\leq\deg s(z)<k+\deg d(z)+1$, similarly,
$$
G_1 \cup G_2=(l(z))+(n+\deg s(z)-\deg d(z)-k-1)P_{\infty}
$$
and
$$
Hull(GRS_k(\boldsymbol \alpha, \boldsymbol v))=C_{\mathscr L}(D,(d(z))+(k+\deg d(z)-\deg s(z)-1)P_{\infty})
$$
is a GRS code.

Moreover, if $k+\varepsilon -\deg d(z)-1<\deg s(z)<k+\deg d(z)+1$, then the hull of $GRS_k(\boldsymbol \alpha, \boldsymbol v)$ is also a GRS code.
\end{proof}
\begin{rem}In Theorem \ref{main th}, we need the condition: $2\deg d(z)+2> \varepsilon$.
\end{rem}

In the following, we always assume the following notations.
\begin{center}
\begin{tabular}{ll}
$s(z)=s_0+s_1z+\cdots+s_{\mu} z^{\mu}\in \Bbb F_q[z]$ with $s_{\mu}\ne 0$. \\
$t(z)=t_0+t_1z+\cdots+t_{\nu}z^{\nu}\in \Bbb F_q[z]$ with $t_{\nu}\ne 0$. \\
$d(z)=\gcd(s(z),t(z))\in \Bbb F_q[z]$ is the greatest common divisor of $s(z)$ and $t(z)$.\\
$l(z)=\mbox{lcm}(s(z),t(z))\in \Bbb F_q[z]$ is the least common multiply of $s(z)$ and $t(z)$.\\
$\deg s(z)=\mu$ and $\deg t(z)=\nu$, where $\mu+\nu=n+\varepsilon$ and $\varepsilon$ is defined by (\ref{v}).\\
$\deg d(z)=\delta$.\\
 \end{tabular}
\end{center}

  If the conditions in Theorem \ref{main th} are not satisfied, then  we shall give a method of linear algebra to find the bases of hulls of GRS codes and exactly compute the dimensions of the hulls.

   If the conditions in Theorem \ref{main th} are not satisfied, then one of the following three statements holds:

 (1) $2\deg d(z)+2>\varepsilon$ and $\deg s(z)\leq k+\varepsilon-\deg d(z)-1$.

 (2) $2\deg d(z)+2>\varepsilon$ and $\deg s(z)\geq k+\deg d(z)+1.$

 (3) $2\deg d(z)+2\le \varepsilon$.

Now we shall consider the above three cases.
Recall that $GRS_k(\boldsymbol \alpha, \boldsymbol v)$ and $GRS_k(\boldsymbol \alpha, \boldsymbol v)^\perp$ have generator matrices:
$$
\begin{pmatrix}
\frac 1{s(\alpha_1)} & \frac 1{s(\alpha_2)} & \cdots & \frac 1{s(\alpha_n)} \\
\frac {\alpha_1}{s(\alpha_1)} & \frac{\alpha_2}{s(\alpha_2)} & \cdots & \frac{\alpha_n}{s(\alpha_n)} \\
\vdots & \vdots & & \vdots \\
\frac{\alpha_1^{k-1}}{s(\alpha_1)} & \frac{\alpha_2^{k-1}}{s(\alpha_2)} & \cdots & \frac{\alpha_n^{k-1}}{s(\alpha_n)}
\end{pmatrix}_{k\times n}
~~\text{and}~~
\begin{pmatrix}
\frac 1{t(\alpha_1)} & \frac 1{t(\alpha_2)} & \cdots & \frac 1{t(\alpha_n)} \\
\frac {\alpha_1}{t(\alpha_1)} & \frac{\alpha_2}{t(\alpha_2)} & \cdots & \frac{\alpha_n}{t(\alpha_n)} \\
\vdots & \vdots & & \vdots \\
\frac{\alpha_1^{n-k-1}}{t(\alpha_1)} & \frac{\alpha_2^{n-k-1}}{t(\alpha_2)} & \cdots & \frac{\alpha_n^{n-k-1}}{t(\alpha_n)}
\end{pmatrix}_{(n-k)\times n},
$$
respectively, where
$$s(z)=s_0+s_1z+\cdots+s_{\mu} z^{\mu}, t(z)=t_0+t_1z+\cdots+t_{\nu}z^{\nu}\in \Bbb F_q[z],
$$
and $s(z)t(z)=u(z)h(z)+h'(z)$, $ s_{\mu} t_{\nu}\neq 0$,  $\deg s(z)+\deg t(z)=n+\varepsilon$.

Now we find two polynomials $$f(z)=f_0+f_1z+\cdots+f_{k-1}z^{k-1}, g(z)=g_0+g_1z+\cdots+g_{n-k-1}z^{n-k-1}\in \Bbb F_q[z]$$ such that
$$
\begin{pmatrix}
\frac{f(\alpha_1)}{s(\alpha_1)},\frac{f(\alpha_2)}{s(\alpha_2)},\ldots,
\frac{f(\alpha_n)}{s(\alpha_n)}
\end{pmatrix}
=\begin{pmatrix}
\frac{g(\alpha_1)}{t(\alpha_1)},\frac{g(\alpha_2)}{t(\alpha_2)},\ldots,
\frac{g(\alpha_n)}{t(\alpha_n)}
\end{pmatrix}
\in Hull(GRS_k(\boldsymbol \alpha, \boldsymbol v)),
$$
which is equivalent to
$$
\frac{f(\alpha_i)}{s(\alpha_i)}=\frac{g(\alpha_i)}{t(\alpha_i)},
f(\alpha_i)t(\alpha_i)=g(\alpha_i)s(\alpha_i), i=1,\ldots,n,
$$
which is equivalent to
\begin{equation*}
f(z)t(z)-g(z)s(z)\equiv 0 \pmod {h(z)},
\end{equation*}
 which  is equivalent to
\begin{equation} \label{mod h}
f(z)t(z)-g(z)s(z)=r(z) h(z),
\end{equation}
where $r(z)\in \Bbb F_q[z]$ and $\deg r(z)$ depends on the higher degree of either  $f(z)t(z)$ or $g(z)s(z)$.

In fact, $\deg(f(z)t(z))\le k-1+\deg t(z)=k-1+n+\varepsilon-\deg s(z)$, $\deg(g(z)s(z))\le n-k-1+\deg s(z)$. It is clear that $n-k-1+\deg s(z) <k-1+n+\varepsilon-\deg s(z)$ if and only if $\deg s(z)<k+\frac{\varepsilon}{2}$. Hence
\begin{itemize}
\item If $\deg s(z)<k+\frac{\varepsilon}{2}$, then $\deg r(z)\le k+\varepsilon -\deg s(z)-1$ in (\ref{mod h}).
\item If $\deg s(z)\ge k+\frac{\varepsilon}{2}$, then $\deg r(z)\le\deg s(z)-k-1$ in (\ref{mod h}).
\end{itemize}

Now we consider the first case
$$2\deg d(z)+2> \varepsilon  \quad\mbox{and}\quad \deg s(z)\leq k+\varepsilon-\deg d(z)-1.$$
 Then the conditions in Theorem \ref{main th} are not satisfied,  $\deg s(z)<k+\frac{\varepsilon}{2}$, and $\deg r(z)\le k+\varepsilon-\deg s(z)-1$ in (\ref{mod h}). To  find  $f(z)$ and $g(z)$  in (\ref{mod h}), we consider  two cases. \\

\textbf{Case 1.} $f(z)t(z)=g(z)s(z)$.

If  $f(z)=\sum_{i=0}^{k-1}f_iz^i$ and $g(z)=\sum_{j=0}^{n-k-1}g_jz^j\in\Bbb F_q[z]$ satisfy $f(z)t(z)=g(z)s(z)$, then $\frac {f(z)}{s(z)}=\frac{g(z)}{t(z)}\in \mathscr L(G_1)\cap\mathscr L(G_2)$. By the proof of Theorem \ref{main th} and $\deg s(z)\leq k+\varepsilon-\deg d(z)-1$, $\deg(G_1\cap G_2)\le -2$ and $\mathscr L(G_1)\cap\mathscr L(G_2)=\{0\}$. Hence $f(z)=g(z)=0$.

On the other hand, the equality  $t(z)f(z)=s(z)g(z)$ induces that
 the  system of equations has only zero solution:$$A\left(\begin{array}{cccccccc}f_0&f_1&\cdots&f_{k-1}&-g_0&-g_1&\cdots&
 -g_{n-k-1}\end{array}\right)^{\mathrm T}=\boldsymbol 0,$$ where
\begin{equation}\label{A}
A=\begin{array}{ll}
&\begin{array}{c@{\hspace{-5pt}}l}
  \left(
 \begin{array}{cccc|cccc}
t_0 & 0 & \cdots &  0    & s_0 & 0 &  \cdots & 0 \\
t_1 & t_0 &  & \vdots & s_1 & s_0 &  & \vdots\\
\vdots &t_1 &  & 0 & \vdots & s_1 & & s_0\\
t_{\nu} & \vdots &  & t_0 & s_\mu &\vdots & & s_1\\
0 & t_{\nu} &   &  t_1    & 0 & s_\mu &  & \vdots \\
0 & 0 &   & \vdots & 0 &0 &  &s_\mu\\
\vdots & \vdots &  & \vdots &\vdots & \vdots& & \vdots\\
0 & 0 & \cdots & t_{\nu} &0  & 0 & \cdots & 0
 \end{array}
 \right)
 \end{array}
\\
 &\begin{array}{ccc}
\quad\underbrace{\rule{25mm}{0mm}}_{k~~\text{columns}} & &\underbrace{\rule{25mm}{0mm}}_{{n-k}~~\text{columns}}
 \end{array}
\end{array}=(A_1,\ldots, A_n)
\end{equation}
is an $(n+\varepsilon+k-\mu)\times n$ coefficient matrix, each $A_i$ is the $i$-th column vector of $A$, and $\mathrm T$ is a transpose transformation. Then    $\mbox{rank}(A)=n$ because the above  system has only zero solution.

\textbf{Case 2.} $f(z)t(z)-g(z)s(z)=r(z)h(z)$, where $0\ne r(z)\in \Bbb F_q[z]$ and  $\deg r(z)\leq k+\varepsilon-\deg s(z)-1$.

By $d(z)=\gcd(s(z),t(z))$ and $\gcd(h(z), s(z)t(z))=1$, $d(z)\mid r(z)$ and $r(z)=d(z)\overline r(z)$. Hence
\begin{equation}\label{g}
 f(z)t(z)-g(z)s(z)=\overline r(z)d(z)h(z),
\end{equation}
where  $0\le \deg \overline r(z)\leq  k+\varepsilon -\deg s(z)-\deg d(z)-1:=\gamma.$  Let
$d(z)h(z)=\sum_{i=0}^{n+\delta}b_iz^i$ and
\begin{equation}\label{b}
B_i=(\underbrace{0,\ldots, 0}_{i\ 0's}, b_0,b_1,\ldots, b_{n+\delta}, 0,\ldots,0)^{\mathrm T} \in \Bbb F_q^{n+\varepsilon+k-\mu},\end{equation}
 where $i=0, 1,\ldots, \gamma=k+\varepsilon -\deg s(z)-\deg d(z)-1$.
  Set
 \begin{equation}\label{B}
 B=(B_0,\ldots, B_{\gamma})
 \end{equation}
to be an $(n+\varepsilon+k-\mu)\times (\gamma+1)$ matrix,  where each $B_i$ is defined by (\ref{b}).

By {\bf Case 1}, if there is   a  polynomial $\overline r(z)$  satisfying (\ref{g}), then $(f(z), g(z))$ must be  unique.
If there is a polynomial $\overline r(z)=\sum_{i=0}^{\gamma}c_iz^i$ such that  a unique pair $(f(z)=\sum_{i=0}^{k-1}f_iz^i, g(z)=\sum_{i=0}^{n-k-1}g_iz^i)$ satisfies (\ref{g}), then $$\sum_{i=0}^{\gamma}c_iB_i=\sum_{i=0}^{k-1}f_iA_{i+1}-\sum_{i=0}^{n-k-1}g_iA_{i+k+1}\in L(A_1,\ldots, A_n),$$
where $L(A_1,\ldots, A_n)$ is a subspace generated by $A_1,\ldots, A_n$. Conversely, if $\sum_{i=0}^{\gamma}c_iB_i\in L(A_1,\ldots, A_n)$, then there is an  polynomial $\overline r(z)=\sum_{i=0}^{\gamma}c_iz^i$ such that  a unique pair $(f(z), g(z))$ satisfies (\ref{g}).
Let \begin{equation}\Omega=\{\overline r(z)=\sum_{i=0}^{\gamma}c_iz^i\in\Bbb F_q[z]: \mbox{$\overline r(z)$ satisfies (\ref{g})}\}\end{equation}
be an $\Bbb F_q$-linear subspace, then
 there is an  isomorphism between two $\Bbb F_q$-linear spaces:
\begin{eqnarray*}\Omega &\longrightarrow& L(A_1, \ldots, A_n)\cap L(B_0,\ldots, B_{\gamma})\\ \overline r(z)=\sum_{i=0}^{\gamma}c_iz^i &\longmapsto&\sum_{i=0}^{\gamma}c_iB_i,
\end{eqnarray*}
and there is an isomorphism between two $\Bbb F_q$-linear spaces:
\begin{eqnarray*}
L(A_1, \ldots, A_n)\cap L(B_0,\ldots, B_{\gamma})&\longrightarrow& Hull(GRS_k(\boldsymbol \alpha, \boldsymbol v))\\
\sum_{i=0}^{\gamma}c_iB_i &\longmapsto&(\frac{f(\alpha_1)}{s(\alpha_1)},\ldots, \frac{f(\alpha_n)}{s(\alpha_n)}),
\end{eqnarray*}
where$\sum_{i=0}^{\gamma}c_iB_i=\sum_{i=0}^{k-1}f_iA_{i+1}-\sum_{i=0}^{n-k-1}g_iA_{i+k+1}$ and $(\frac{f(\alpha_1)}{s(\alpha_1)},\ldots, \frac{f(\alpha_n)}{s(\alpha_n)})= (\frac{g(\alpha_1)}{t(\alpha_1)},\ldots,\frac{g(\alpha_n)}{t(\alpha_n)})$.

In the following, we choose all  non-zero  polynomials of distinct degrees in $\Omega$
\begin{equation}\label{r}
\overline r_1(z), \overline r_2(z),\ldots, \overline r_e(z),
\end{equation}
where $0\le \deg \overline r_1(z)<\deg \overline r_2(z)<\cdots<\deg \overline r_e(z)\le \gamma=k+\varepsilon-\deg s(z)-\deg d(z)-1$ and each $\overline r_i(z)$ is a polynomial of the $i$-th lowest degree in $\Omega$. We shall prove that (\ref{r}) is a basis of $\Omega$.

It is clear that the elements in (\ref{r}) are linearly independent over $\Bbb F_q$. Now we proceed by induction on the degree of every polynomial $\overline r(z)\in\Omega$ and prove that it can be linearly expressed by $\{\overline r_1(z),\ldots, \overline r_e(z)\}$.

If $\deg \overline r(z)=\deg \overline r_1(z)$, then there is $a\in \Bbb F_q^*$ such that $\overline r(z)-a\overline r_1(z)\in \Omega$ with $\deg (\overline r(z)-a\overline r_1(z))< \deg \overline r_1(z)$.
Since $\overline r_1(z)$ is a polynomial of the lowest degree in $\Omega$, $\overline r(z)=a\overline r_1(z)$.
We assume that the statement holds for every polynomial in $\Omega$ of degree less than $\deg \overline r_d(z)$, $d\le e$. Now let $\overline r(z)\in \Omega$ with $\deg \overline r(z)=\deg \overline r_d(z)$, then there is $a\in \Bbb F_q^*$ such that $\overline r(z)-a\overline r_d(z)\in \Omega$ with $\deg(\overline r(z)-a\overline r_d(z))<\deg \overline r_d(z)$. The induction hypothesis immediately yields the desired result. So (\ref{r}) is a basis of $\Omega$.

By the above two isomorphisms, we can obtain that $f_1(z), \ldots, f_e(z)$ are corresponding to $\overline r_1(z),\ldots, \overline r_e(z)$, respectively, so there is a basis of   $Hull(GRS_k(\boldsymbol \alpha, \boldsymbol v))$ as follows:
  $$(\frac{f_i(\alpha_1)}{s(\alpha_1)},\ldots, \frac{f_i(\alpha_n)}{s(\alpha_n)}), i=1,\ldots, e.$$

  Now  we have an algorithm to find a basis of $Hull(GRS_k(\boldsymbol \alpha, \boldsymbol v))$.
\begin{alg}\label{alg}
If $2\deg d(z)+2>\varepsilon$ and $\deg s(z)\leq k+\varepsilon-\deg d(z)-1$, then
we consider the matrix $$(A|B)=(A_1,\ldots, A_n| B_0,\ldots, B_{\gamma}),$$
where $A$ and $B$ are defined by (\ref{A}) and (\ref{B}), respectively.

(1) If there is the first smallest number $i_1$, $0\le i_1\le \gamma$,  such that  columns $A_1, \ldots, A_n$, $B_0, \ldots, B_{i_1}$ are linearly dependent over $\Bbb F_q$, then there is $(\overline r_1(z), f_1(z), g_1(z))$ satisfying (\ref{g}).

(2) If there is the second smallest number $i_2$, $i_1<i_2\le \gamma$, such that columns $A_1, \ldots, A_n$, $B_0,\ldots, B_{i_1-1}, B_{i_1+1},\ldots, B_{i_2}$ are linearly dependent over $\Bbb F_q$, then there is $(\overline r_2(z), f_2(z), g_2(z))$ satisfying (\ref{g}).

(3) And so on, there are polynomials  $f_1(z), f_2(z), \ldots, f_e(z)\in \Bbb F_q[z]$ such that
$$\left\{\left(\frac{f_i(\alpha_1)}{s(\alpha_1)},\ldots, \frac{f_i(\alpha_n)}{s(\alpha_n)}\right)~\Big|~ i=1,\ldots, e\right\}$$
is a basis of  $Hull(GRS_k(\boldsymbol \alpha, \boldsymbol v))$.
\end{alg}

\begin{thm} \label{main th2}
Let $GRS_k(\boldsymbol \alpha, \boldsymbol v)=C_{\mathscr L}(D,(k-1)P_{\infty}+(s(z)))$ and $GRS_k(\boldsymbol \alpha, \boldsymbol v)^\perp=C_{\mathscr L}(D,(n-k-1)P_{\infty}+(t(z)))$ be the two GRS codes in Propositions \ref{GRS AG} and \ref{GRS dual}, respectively.
If $2\deg d(z)+2>\varepsilon$ and $\deg s(z)\leq k+\varepsilon-\deg d(z)-1$, then
$$
\dim (Hull(GRS_k(\boldsymbol \alpha, \boldsymbol v)))=n+\gamma+1-\mbox{rank}~(A|B),$$
where $\gamma=k+\varepsilon-\deg s(z)-\deg d(z)-1$, $A$ and $B$ are defined by (\ref{A}) and (\ref{B}), respectively.
\end{thm}

If $2\deg d(z)+2>\varepsilon$ and $\deg s(z)\ge k+\deg d(z)+1$ or $2\deg d(z)+2\le\varepsilon$, we have similar formulas for the dimensions of hulls of GRS codes.

\begin{thm}
Let $GRS_k(\boldsymbol \alpha, \boldsymbol v)=C_{\mathscr L}(D,(k-1)P_{\infty}+(s(z)))$ and $GRS_k(\boldsymbol \alpha, \boldsymbol v)^\perp=C_{\mathscr L}(D,(n-k-1)P_{\infty}+(t(z)))$ be the two GRS codes in Propositions \ref{GRS AG} and \ref{GRS dual}, respectively.

(1) If $2\deg d(z)+2>\varepsilon$ and $\deg s(z)\ge k+\deg d(z)+1$, then
$$
\dim (Hull(GRS_k(\boldsymbol \alpha, \boldsymbol v)))=n+\gamma'+1-\mbox{rank}~(A'|B'),$$
where $\gamma'=\deg s(z)-k-\deg d(z)-1$, $A'$ and $B'$ are defined as follows.
\begin{equation}\label{A'}
A'=\begin{array}{ll}
&\begin{array}{c@{\hspace{-5pt}}l}
  \left(
 \begin{array}{cccc|cccc}
t_0 & 0 & \cdots &  0    & s_0 & 0 &  \cdots & 0 \\
t_1 & t_0 &  & \vdots & s_1 & s_0 &  & \vdots\\
\vdots &t_1 &  & t_0 & \vdots & s_1 & & 0\\
t_{\nu} & \vdots &  & t_1 & s_\mu &\vdots & & s_0\\
0 & t_{\nu} &   & \vdots     & 0 & s_\mu &  & s_1 \\
0 & 0 &   & t_{\nu} & 0 &0 &  &\vdots\\
\vdots & \vdots &  & \vdots &\vdots & \vdots& & \vdots\\
0 & 0 & \cdots & 0 &0  & 0 & \cdots & s_\mu
 \end{array}
 \right)
 \end{array}
\\
 &\begin{array}{ccc}
\quad\underbrace{\rule{25mm}{0mm}}_{k~~\text{columns}} & &\underbrace{\rule{25mm}{0mm}}_{{n-k}~~\text{columns}}
 \end{array}
\end{array}=(A'_1,\ldots, A'_n)
\end{equation}
is an $(n+\mu-k)\times n$ matrix, where each $A'_i$ is the $i$-th column vector of $A'$; and
\begin{equation}\label{B'}
 B'=(B'_0,\ldots, B'_{\gamma'})
 \end{equation}
is an $(n+\mu-k)\times (\gamma'+1)$ matrix,  where  $d(z)h(z)=\sum_{i=0}^{n+\delta}b_iz^i$ and each $B'_i$ is the $i$-th column vector of $B'$ defined by
\begin{equation}\label{b'}
B'_i=(\underbrace{0,\ldots, 0}_{i\ 0's}, b_0,b_1,\ldots, b_{n+\delta}, 0,\ldots,0)^{\mathrm T} \in \Bbb F_q^{n+\mu-k},\end{equation}
$i=0, 1,\ldots, \gamma'=\deg s(z)-k-\deg d(z)-1$.

(2) If $2\deg d(z)+2\le \varepsilon$ and $\deg s(z)<k+\frac{\varepsilon}{2}$, then
$$
\dim (Hull(GRS_k(\boldsymbol \alpha, \boldsymbol v)))=n+\gamma+1-\mbox{rank}~(A|B),$$
where $\gamma=k+\varepsilon-\deg s(z)-\deg d(z)-1$, $A$ and $B$ are defined by (\ref{A}) and (\ref{B}), respectively.

(3) If $2\deg d(z)+2\le \varepsilon$ and $\deg s(z)\ge k+\frac{\varepsilon}{2}$, then
$$
\dim (Hull(GRS_k(\boldsymbol \alpha, \boldsymbol v)))=n+\gamma'+1-\mbox{rank}~(A'|B'),$$
where $\gamma'=\deg s(z)-k-\deg d(z)-1$, $A'$ and $B'$ are defined by (\ref{A'}) and (\ref{B'}), respectively.
\end{thm}
\begin{proof}
(1) %Similarly, we consider {\bf Case 1} and {\bf Case 2}.
By the proof of Theorem \ref{main th} and $\deg s(z)\ge k+\deg d(z)+1$,  $\deg(G_1\cap G_2)\le -2$ and $\mathscr L(G_1)\cap\mathscr L(G_2)=\{0\}$.
On the other hand, by $2\deg d(z)+2>\varepsilon$ and $\deg s(z)\ge k+\deg d(z)+1$, $\deg s(z)> k+\frac{\varepsilon}{2}$ and $\deg r(z)\le \deg s(z)-k-1$ in (\ref{mod h}).

Hence the equality  $t(z)f(z)=s(z)g(z)$ induces that the  system of equations has only zero solution:$$A'\left(\begin{array}{cccccccc}f_0&f_1&\cdots&f_{k-1}&-g_0&-g_1&\cdots&
 -g_{n-k-1}\end{array}\right)^{\mathrm T}=\boldsymbol 0,$$
 where $A'$ is an $(n+\mu-k)\times n$ matrix defined by (\ref{A'}). Then $\mbox{rank}(A')=n$ because the above system has only zero solution.

Let $r(z)=d(z)\overline r(z)$ in (\ref{mod h}). Then
$$
 f(z)t(z)-g(z)s(z)=\overline r(z)d(z)h(z).
$$

By $\deg r(z)\le \deg s(z)-k-1$, $0\le \deg \overline r(z)=\deg r(z)-\deg d(z)\le \deg s(z)-k-\deg d(z)-1=\gamma'$. Then let $d(z)h(z)=\sum_{i=0}^{n+\delta}b_iz^i$,
$$
B'_i=(\underbrace{0,\ldots, 0}_{i\ 0's}, b_0,b_1,\ldots, b_{n+\delta}, 0,\ldots,0)^{\mathrm T} \in \Bbb F_q^{n+\mu-k},
$$
where $i=0, 1,\ldots, \gamma'=\deg s(z)-k-\deg d(z)-1$,
and
$$
B'=(B'_0,\ldots, B'_{\gamma'})
$$
 an $(n+\mu-k)\times (\gamma'+1)$ matrix, where each $B'_i$ is the $i$-th column vector of $B'$.
The remainder of the proof is omitted because of the similarity.

(2) If $2\deg d(z)+2\le \varepsilon$ and $\deg s(z)<k+\frac{\varepsilon}{2}$, then $\deg s(z)<k+\varepsilon-\deg d(z)-1$ and by the proof of Theorem \ref{main th}, $\deg(G_1\cap G_2)=\deg d(z)+\deg s(z)-\varepsilon-k-1< -2$ and $\mathscr L(G_1)\cap\mathscr L(G_2)=\{0\}$. On the other hand, by $\deg s(z)< k+\frac{\varepsilon}{2}$, $\deg r(z)\le k+\varepsilon -\deg s(z)-1$ in (\ref{mod h}).

Hence the equality  $t(z)f(z)=s(z)g(z)$ induces that the  system of equations has only zero solution:$$A\left(\begin{array}{cccccccc}f_0&f_1&\cdots&f_{k-1}&-g_0&-g_1&\cdots&
 -g_{n-k-1}\end{array}\right)^{\mathrm T}=\boldsymbol 0,$$
 where $A$ is an $(n+\varepsilon +k-\mu)\times n$ matrix defined by (\ref{A}). Then $\mbox{rank}(A)=n$ because the above  system has only zero solution.

Let $r(z)=d(z)\overline r(z)$ in (\ref{mod h}). Then
$$
 f(z)t(z)-g(z)s(z)=\overline r(z)d(z)h(z).
$$

By $\deg r(z)\le k+\varepsilon -\deg s(z)-1$, $0\le \deg \overline r(z)=\deg r(z)-\deg d(z)\le k+\varepsilon -\deg s(z)-\deg d(z)-1=\gamma$. Then let $d(z)h(z)=\sum_{i=0}^{n+\delta}b_iz^i$ and
$$
B=(B_0,\ldots, B_{\gamma})
$$
 an $(n+\varepsilon +k-\mu)\times (\gamma+1)$ matrix, where each $B_i$ is the $i$-th column vector of $B$ defined by (\ref{b}).
The remainder of the proof is omitted because of the similarity.

(3) The proof of (3) is similar to the proof of (2), so we omit it.
\end{proof}

If $2\deg d(z)+2>\varepsilon$ and $\deg s(z)\ge k+\deg d(z)+1$, we can find a basis of $Hull(GRS_k(\boldsymbol \alpha, \boldsymbol v))$ by the following algorithm.
\begin{alg}\label{alg2}
If $2\deg d(z)+2>\varepsilon$ and $\deg s(z)\ge k+\deg d(z)+1$, then
we consider the matrix $$(A'|B')=(A'_1,\ldots, A'_n| B'_0,\ldots, B'_{\gamma'}),$$
where $\gamma'=\deg s(z)-k-\deg d(z)-1$, $A'$ and $B'$ are defined by (\ref{A'}) and (\ref{B'}), respectively.

(1) If there is the first smallest number $i_1$, $0\le i_1\le \gamma'$,  such that  columns $A'_1, \ldots, A'_n$, $B'_0, \ldots, B'_{i_1}$ are linearly dependent over $\Bbb F_q$, then there is $(\overline r_1(z), f_1(z), g_1(z))$ satisfying (\ref{g}).

(2) If there is the second smallest number $i_2$, $i_1<i_2\le \gamma'$, such that columns $A'_1, \ldots, A'_n$, $B'_0,\ldots, B'_{i_1-1}, B'_{i_1+1},\ldots, B'_{i_2}$ are linearly dependent over $\Bbb F_q$, then there is $(\overline r_2(z), f_2(z), g_2(z))$ satisfying (\ref{g}).

(3) And so on, there are polynomials  $f_1(z), f_2(z), \ldots, f_e(z)\in \Bbb F_q[z]$ such that
$$\left\{\left(\frac{f_i(\alpha_1)}{s(\alpha_1)},\ldots, \frac{f_i(\alpha_n)}{s(\alpha_n)}\right)~\Big|~ i=1,\ldots, e\right\}$$
is a basis of  $Hull(GRS_k(\boldsymbol \alpha, \boldsymbol v))$.
\end{alg}

\begin{rem}
If $2\deg d(z)+2\le \varepsilon$ and $\deg s(z)<k+\frac{\varepsilon}{2}$, then the algorithm for finding a basis of $Hull(GRS_k(\boldsymbol \alpha, \boldsymbol v))$ is just Algorithm \ref{alg}.

If $2\deg d(z)+2\le \varepsilon$ and $\deg s(z)\ge k+\frac{\varepsilon}{2}$, then the algorithm for finding a basis of $Hull(GRS_k(\boldsymbol \alpha, \boldsymbol v))$ is just Algorithm \ref{alg2}.
\end{rem}

\begin{cor}\label{lcd}
Let $GRS_k(\boldsymbol \alpha, \boldsymbol v)=C_{\mathscr L}(D,(k-1)P_{\infty}+(s(z)))$ and $GRS_k(\boldsymbol \alpha, \boldsymbol v)^\perp=C_{\mathscr L}(D,(n-k-1)P_{\infty}+(t(z)))$ be the two GRS codes in Propositions \ref{GRS AG} and \ref{GRS dual}, respectively.

If $2\deg d(z)+2>\varepsilon$ and $\deg s(z)\leq k+\varepsilon-\deg d(z)-1$ or $2\deg d(z)+2\le \varepsilon$ and $\deg s(z)<k+\frac{\varepsilon}{2}$, and the rank of the matrix $(A|B)$ is equal to $n+\gamma +1$, where $A$ and $B$ are defined by (\ref{A}) and (\ref{B}), respectively, then $Hull(GRS_k(\boldsymbol \alpha, \boldsymbol v))$ is an LCD code.

If $2\deg d(z)+2>\varepsilon$ and $\deg s(z)\ge k+\deg d(z)+1$ or $2\deg d(z)+2\le \varepsilon$ and $\deg s(z)\ge k+\frac{\varepsilon}{2}$, and the rank of the matrix $(A'|B')$ is equal to $n+\gamma'+1$, where $A'$ and $B'$ are defined by (\ref{A'}) and (\ref{B'}), respectively, then $Hull(GRS_k(\boldsymbol \alpha, \boldsymbol v))$ is an LCD code.
\end{cor}

\begin{rem}
Corollary \ref{lcd} shows that $Hull(GRS_k(\boldsymbol \alpha, \boldsymbol v))$  may be an LCD code if the conditions in Theorem \ref{main th} are not satisfied. Then there is a question:

If the conditions in Theorem \ref{main th} are not satisfied, is it possible that the hull of $GRS_k(\boldsymbol \alpha, \boldsymbol v)$ is a self-orthogonal code or a self-dual code?
\end{rem}

In the following, we shall provide some examples to show that the conditions that hulls of GRS codes are also GRS codes in Theorem \ref{main th} are good enough and can not be changed. It means that if the conditions in Theorem \ref{main th} are not satisfied, then hulls of GRS codes may be GRS codes or not GRS codes.
\begin{exa}\label{u=0,eg1}
Let $q=3^2$, $\Bbb F_q^*=\langle \gamma \rangle$, and $h(z)=z^7 + \gamma z^6 + \gamma^2z^5 + \gamma^3z^4 + 2z^3 + \gamma^5z^2 + \gamma^6z + \gamma^7\in \Bbb F_q[z]$.
Then $h'(z)=z^6 + \gamma^6 z^4 + \gamma^3 z^3 + \gamma z + \gamma^6
=(z^3 + \gamma^3z^2 + \gamma^5z + \gamma)(z^3+\gamma^7 z^2 + \gamma^3z + \gamma^5)\in \Bbb F_q[z]$ is an irreducible factorization over $\Bbb F_q$.

Let $s(z)=z^3 + \gamma^3z^2 + \gamma^5z + \gamma$ and $t(z)=z^3+\gamma^7 z^2 + \gamma^3z + \gamma^5\in \Bbb F_q[z]$. Then $\gcd(s(z),t(z))=1$. Let $\boldsymbol \alpha=(\alpha_1,\alpha_2,\ldots,\alpha_7)=
(1,\gamma^2,\gamma^3,\gamma^4,\gamma^5,\gamma^6,\gamma^7)$ and $\boldsymbol v=(v_1,v_2,\ldots,v_7)=(\gamma^2, \gamma^3, 1, \gamma^7, \gamma^3, \gamma^5, \gamma)$, where $\alpha_i$ are the roots of $h(z)$ and $v_i=(s(\alpha_i))^{-1},i=1,\ldots,7$. Then $GRS_5(\boldsymbol \alpha, \boldsymbol v)$ is a $[7,5,3]$ linear code, the column of the matrix $(A|B)$ is full rank, where $A$ and $B$ are defined by (\ref{A}) and (\ref{B}), respectively,
 and $Hull(GRS_k(\boldsymbol \alpha,\boldsymbol v))=\{\boldsymbol0\}$ by Magma \cite{m}.
\end{exa}
\begin{exa}\label{u=0,eg2}
Let $q=2^4$, $\Bbb F_q^*=\langle \gamma \rangle$, and $h(z)=z^{11} + \gamma^{13} z^{10} + \gamma^6 z^8 + \gamma^5z^7 + \gamma^7z^6 + \gamma^8z^5 + \gamma z^3 + z^2 + \gamma z\in \Bbb F_q[z]$. Then $h'(z)=z^{10} + \gamma^5 z^6 + \gamma^8z^4 + \gamma z^2 + \gamma=(z + \gamma^9)^2(z^2 + \gamma^8 z + \gamma^8)^2(z^2 + \gamma^{12}z + \gamma^6)^2\in \Bbb F_q[z]$ is an irreducible factorization over $\Bbb F_q$.

Let $s(z)=(z + \gamma^9)^2(z^2 + \gamma^8 z + \gamma^8)^2$ and $t(z)=(z^2 + \gamma^{12}z + \gamma^6)^2\in \Bbb F_q[z]$. Then $\gcd(s(z),t(z))=1$. Let $\boldsymbol \alpha=(\alpha_1,\alpha_2,\ldots,\alpha_{11})=
(\gamma^2,\gamma^3,\gamma^4,\gamma^5,\gamma^6,\gamma^8,\gamma^{10},
\gamma^{11},\gamma^{13},\gamma^{14},0)$ and $\boldsymbol v=(v_1,v_2,\ldots,v_{11})=(\gamma^8, \gamma^8, \gamma^8, \gamma^9, \gamma^{12}, \gamma^5, \gamma^5, 1, \gamma^5, \gamma^{14}, \gamma^{11})$, where $\alpha_i$ are the roots of $h(z)$ and $v_i=(s(\alpha_i))^{-1},i=1,\ldots,11$. Then $GRS_8(\boldsymbol \alpha, \boldsymbol v)$ is a $[11,8,4]$ linear code and $Hull(GRS_8(\boldsymbol \alpha,\boldsymbol v))$ is a $[11,1,10]$ linear code with a generator matrix
$$
\left(\begin{array}{ccccccccccc}
1 & \gamma^4 & 0 & \gamma^{12}& 1& \gamma^9& \gamma^9& \gamma^{14}& \gamma^{11}& \gamma^{14}& \gamma
\end{array}\right).
$$
\end{exa}
\begin{rem}
Note that in Examples \ref{u=0,eg1} and \ref{u=0,eg2}, $u(z)=0$ and the conditions in Theorem \ref{main th} are not satisfied. Moreover, Example \ref{u=0,eg2} shows that $Hull(GRS_8(\boldsymbol \alpha,\boldsymbol v))$ is not a GRS code as the third component of the $1$-dimensional generator matrix is zero, which further shows that the conditions in Theorem \ref{main th} are good enough.
\end{rem}

%If the conditions $\deg u(z)\leq 2\deg d(z)$ and $k+\deg u(z)-\deg d(z)-1<\deg s(z)<k+\deg d(z)+1$ in Theorem \ref{main th} are not satisfied, then the bound (\ref{bound}) about the hull dimension of GRS codes also holds. Specifically, we have the following theorem similar to Theorem \ref{case 2 th ns} and we omit the proof here.
\begin{exa}\label{eg1}
Let $q=5^2$, $\Bbb F_q^*=\langle \gamma \rangle$, and $h(z)=z^7 + z^6 + \gamma^{15}z^5 + \gamma^4z^4 + \gamma^{20}z^3 + \gamma^{13}z^2 + \gamma z + \gamma^{13}\in \Bbb F_q[z]$.
Then $zh(z)+h'(z)=z^8+z^7+\gamma^{17}z^6 + \gamma^{23}z^5 + \gamma^{20}z^4 + \gamma^{23}z^3 +2z^2 + \gamma^7z + \gamma
=(z+\gamma)(z+\gamma^2)(z+\gamma^{22})(z^5 + \gamma^4z^4 + z^3 + 4z^2 + \gamma z + 1)\in \Bbb F_q[z]$ is an irreducible factorization over $\Bbb F_q$.

Let $s(z)=(z+\gamma)(z+\gamma^2)(z+\gamma^{22})$ and $t(z)=z^5 + \gamma^4z^4 + z^3 + 4z^2 + \gamma z + 1\in \Bbb F_q[z]$. Then $\gcd(s(z),t(z))=1$. Let $\boldsymbol \alpha=(\alpha_1,\alpha_2,\ldots,\alpha_7)=
(1,\gamma,\gamma^2,\gamma^4,\gamma^5,\gamma^6,\gamma^7)$ and $\boldsymbol v=(v_1,v_2,\ldots,v_7)=(3, \gamma^{10}, \gamma^{20}, \gamma^2, 2, 2, \gamma^4)$, where $\alpha_i$ are the roots of $h(z)$ and $v_i=(s(\alpha_i))^{-1},i=1,\ldots,7$. Then $GRS_3(\boldsymbol \alpha, \boldsymbol v)$ is a $[7,3,5]$ linear code,
 the column of the matrix $(A|B)$ is full rank, where $A$ and $B$ are defined by (\ref{A}) and (\ref{B}), respectively,
 and $Hull(GRS_k(\boldsymbol \alpha,\boldsymbol v))=\{\boldsymbol0\}$ by Magma \cite{m}.
\end{exa}
\begin{exa}\label{eg2}
Let $q=2^4$, $\Bbb F_q^*=\langle \gamma \rangle$, and $h(z)=z^{11} + \gamma^{12}z^8 + \gamma^{11}z^7 + \gamma^9z^5 + \gamma^7z^3 + \gamma^6z^2 + \gamma^5z + \gamma^4 \in \Bbb F_q[z]$. Then $h(z)+h'(z)=z^{11}+z^{10} +\gamma^{12}z^8+\gamma^{11}z^7 + \gamma^{11}z^6 + \gamma^9z^5 +\gamma^9z^4+ \gamma^7z^3+\gamma^{10}z^2 + \gamma^5z+\gamma^8=(z + 1)(z^2 + \gamma^3z + 1)(z^8 + \gamma^3z^7 + \gamma^{13}z^6 + \gamma^8z^5 + \gamma z^4 + \gamma^{14}z^3 + \gamma^4z^2 + \gamma^{13}z +\gamma^8)\in \Bbb F_q[z]$ is an irreducible factorization over $\Bbb F_q$.

Let $s(z)=(z + 1)(z^2 + \gamma^3z + 1)$ and $t(z)=z^8 + \gamma^3z^7 + \gamma^{13}z^6 + \gamma^8z^5 + \gamma z^4 + \gamma^{14}z^3 + \gamma^4z^2 + \gamma^{13}z +\gamma^8\in \Bbb F_q[z]$. Then $\gcd(s(z),t(z))=1$. Let $\boldsymbol \alpha=(\alpha_1,\alpha_2,\ldots,\alpha_{11})=
(\gamma^{11},\gamma^{12},\gamma^2,\gamma^{13},\gamma^{14},\gamma^4,\gamma^{5},
\gamma^{6},\gamma^{8},\gamma^{9},\gamma^{10})$ and $\boldsymbol v=(v_1,v_2,\ldots,v_{11})=(\gamma^{14}, \gamma^{10}, \gamma^5, \gamma^{11}, \\ \gamma^9, \gamma^2, \gamma, 1, 1, \gamma^3, \gamma)$, where $\alpha_i$ are the roots of $h(z)$ and $v_i=(s(\alpha_i))^{-1},i=1,\ldots,11$. Then $GRS_4(\boldsymbol \alpha, \boldsymbol v)$ is a $[11,4,8]$ linear code and $Hull(GRS_4(\boldsymbol \alpha,\boldsymbol v))$ is a $[11,1,10]$ linear code with a generator matrix
$$
\left(\begin{array}{ccccccccccc}
0 & 1 & 1 & \gamma^{10}& \gamma^9& \gamma^{14}& \gamma^8& \gamma^7& \gamma^6& \gamma^{11}&\gamma^{13}
\end{array}\right).
$$
\end{exa}
\begin{rem}
Note that in Examples \ref{eg1} and \ref{eg2}, $u(z)\ne 0$ and the conditions in Theorem \ref{main th} are not satisfied. Moreover, Example \ref{eg2} shows that $Hull(GRS_4(\boldsymbol \alpha,\boldsymbol v))$ is not a GRS code as the first component of the $1$-dimensional generator matrix is zero, which further shows that the conditions in Theorem \ref{main th} are good enough.
\end{rem}
\section{Self-orthogonal and self-dual GRS codes}
In Section 3, Theorem \ref{main th} gives the conditions that hulls of GRS codes are also GRS codes. In this section, we shall give some corollaries of Theorem \ref{main th}. Specifically, we shall provide methods for constructing self-orthogonal or self-dual GRS codes. Besides, we give some examples to support our findings by Magma \cite{m}.

\begin{cor}\label{th appl}
Let $GRS_k(\boldsymbol \alpha, \boldsymbol v)=C_{\mathscr L}(D,(k-1)P_{\infty}+(s(z)))$ and $GRS_k(\boldsymbol \alpha, \boldsymbol v)^\perp=C_{\mathscr L}(D,(n-k-1)P_{\infty}+(t(z)))$ be the two GRS codes in Propositions \ref{GRS AG} and \ref{GRS dual}, respectively. Let $2\deg d(z)+2> \varepsilon$.

(1) If $k+\varepsilon-\deg d(z)-1 < \deg s(z) < k+\frac {\varepsilon}2$ and $\deg s(z)- (k+\varepsilon-\deg d(z)-1)=1$, then $GRS_k(\boldsymbol \alpha, \boldsymbol v)$ is an LCD code; if $k+\frac {\varepsilon}2 \leq \deg s(z) < k+\deg d(z)+1$ and $k+\deg d(z)+1-\deg s(z)=1$, then $GRS_k(\boldsymbol \alpha, \boldsymbol v)$ is an LCD code.

(2) If $k+\frac {\varepsilon}2 \leq \deg s(z)<k+\deg d(z)+1$ and $\gcd(s(z),t(z))=s(z)$, then $GRS_k(\boldsymbol \alpha, \boldsymbol v)$ is a self-orthogonal code. Moreover, $t(z)=\lambda(z) s(z)$ and Equation (\ref{im eq})  is  changed to
\begin{equation} \label{appl eq1}
\lambda(z)s(z)^2=u(z)h(z)+ h'(z),
\end{equation}
 where $\lambda(z)\in \Bbb F_q[z]$ is a polynomial.

 If $k+\varepsilon-\deg d(z)-1 < \deg s(z)< k+\frac {\varepsilon}2$ and $\gcd(s(z),t(z))=t(z)$, then $GRS_k(\boldsymbol \alpha, \boldsymbol v)$ is a dual-containing code.

(3) If $k+\frac {\varepsilon}2 \leq \deg s(z)<k+\deg d(z)+1$, $\gcd(s(z),t(z))=s(z)$, and $n=2k$, then $GRS_k(\boldsymbol \alpha, \boldsymbol v)$ is a self-dual code. Moreover, $t(z)=\lambda s(z)$ and Equation (\ref{im eq}) is changed to
\begin{equation} \label{appl eq2}
 \lambda s(z)^2=u(z)h(z)+ h'(z),
\end{equation}
where $\lambda\in \Bbb F_q^*$, $u(z)\in \Bbb F_q[z]$ is a nonzero polynomial with $0\leq \deg u(z)\leq n-2$ and $\deg u(z)$ is even.
\end{cor}
\begin{proof}
(1) If $k+\varepsilon-\deg d(z)-1 < \deg s(z)<k+\frac {\varepsilon}2$, then Equation (\ref{th2 eq2}) holds by Theorem \ref{main th}.
Moreover, if $\deg s(z)-(k+\varepsilon-\deg d(z)-1)=1$,
$$
\deg((d(z))+(\deg d(z)+\deg s(z)-k-\varepsilon-1)P_{\infty})=\deg d(z)+\deg s(z)-k-\varepsilon-1=-1<0.
$$
Therefore, by Proposition \ref{prop para}, $\dim(Hull(GRS_k(\boldsymbol \alpha, \boldsymbol v)))=0.$
Another result can be proved similarly.

(2) If $k+\frac {\varepsilon}2 \leq \deg s(z)<k+\deg d(z)+1$, then Equation (\ref{th2 eq1}) holds by Theorem \ref{main th}.
Moreover, if $\gcd(s(z),t(z))=s(z)$, then by (\ref{th2 eq1}),
$$
Hull(GRS_k(\boldsymbol \alpha, \boldsymbol v))=C_{\mathscr L}(D,(s(z))+(k-1)P_{\infty})=GRS_k(\boldsymbol \alpha, \boldsymbol v),
$$
which infers that $GRS_k(\boldsymbol \alpha, \boldsymbol v)\subseteq GRS_k(\boldsymbol \alpha, \boldsymbol v)^\perp$, i.e., $GRS_k(\boldsymbol \alpha, \boldsymbol v)$ is a self-orthogonal code. By $\gcd(s(z),t(z))=s(z)$, we can assume that $t(z)=\lambda(z)s(z)$ for some polynomial $\lambda(z)\in \Bbb F_q[z]$, then Equation (\ref{im eq}) can be rewritten as
$$
\lambda(z)s(z)^2=u(z)h(z)+ h'(z).
$$

If $k+\varepsilon-\deg d(z)-1 < \deg s(z)< k+\frac {\varepsilon}2$ and $\gcd(s(z),t(z))=t(z)$, then $GRS_k(\boldsymbol \alpha, \boldsymbol v)^\perp\subseteq GRS_k(\boldsymbol \alpha, \boldsymbol v)$, i.e., $GRS_k(\boldsymbol \alpha, \boldsymbol v)$ is a dual-containing code.

(3) If $k+\frac {\varepsilon}2 \leq \deg s(z)<k+\deg d(z)+1$, $\gcd(s(z),t(z))=s(z)$, and $n=2k$, then $GRS_k(\boldsymbol \alpha, \boldsymbol v)$ is a self-orthogonal code by (2) and
$$
\dim(Hull(GRS_k(\boldsymbol \alpha, \boldsymbol v)))=\dim(GRS_k(\boldsymbol \alpha, \boldsymbol v))=k=\frac{n}{2},
$$
which infers that $GRS_k(\boldsymbol \alpha, \boldsymbol v)$ is a self-dual code. Moreover, by (\ref{appl eq1}),
$$
2\deg s(z)=n+\varepsilon-\deg \lambda(z)\leq 2k+\varepsilon.
$$
On the other hand, $2\deg s(z) \geq 2k+\varepsilon$ by $k+\frac {\varepsilon}2 \leq \deg s(z)$, then $2\deg s(z) = 2k+\varepsilon$ and $\deg \lambda(z)=0$.
Besides, by $\varepsilon=2(\deg s(z)-k)$, $u(z)$ is a nonzero polynomial of an even degree (otherwise, $u(z)=0$ and $\varepsilon=-1$, which is contradictory).
 Therefore, we obtain Equation (\ref{appl eq2}).
\end{proof}
\begin{rem}
Corollary \ref{th appl} shows self-orthogonal and self-dual GRS codes. Specifically, we can select polynomials $h(z),u(z)\in \Bbb F_q[z]$ with $\deg u(z)\leq n-2$ such that
$$
u(z)h(z)+ h'(z)=\lambda(z)s(z)^2
$$
for some polynomials $\lambda(z)$ and $s(z)\in \Bbb F_q[z]$ with $\deg s(z)\geq k+\frac{\varepsilon}{2}$. Then let $\boldsymbol \alpha=(\alpha_1,\ldots,\alpha_n)$, where $\alpha_i$ are the roots of $h(z)$, $i=1,\ldots,n$. Let $\boldsymbol v=(v_1,\ldots,v_n)$, where $v_i=(s(\alpha_i))^{-1},i=1,\ldots,n$. By Corollary \ref{th appl} (2), $GRS_k(\boldsymbol \alpha, \boldsymbol v)$ is a self-orthogonal code.

On the other hand, we can select polynomials $h(z),u(z)\in \Bbb F_q[z]$ with $0\leq \deg u(z)\leq n-2$  such that
$$
u(z)h(z)+ h'(z)=\lambda s(z)^2
$$
for $\lambda\in \Bbb F_q^*$ and some polynomial $s(z)\in \Bbb F_q[z]$ with $\deg s(z)\geq k+\frac{\varepsilon}{2}$. Then let $\boldsymbol \alpha=(\alpha_1,\ldots,\alpha_n)$, where $\alpha_i$ are the roots of $h(z)$, $i=1,\ldots,n$. Let $\boldsymbol v=(v_1,\ldots,v_n)$, where $v_i=(s(\alpha_i))^{-1},i=1,\ldots,n$. By Corollary \ref{th appl} (3), $GRS_k(\boldsymbol \alpha, \boldsymbol v)$ is a self-dual code.
\end{rem}

\begin{cor}
In (\ref{im eq}), let $u(z)=0$. If $k-\deg d(z)-2<\deg s(z)<k+\deg d(z)+1$, then the hull of $GRS_k(\boldsymbol \alpha, \boldsymbol v)$ is also a GRS code. Moreover,
if $k-\deg d(z)-2 < \deg s(z)< k$, then
$$
Hull(GRS_k(\boldsymbol \alpha, \boldsymbol v))=C_{\mathscr L}(D,(d(z))+(\deg d(z)+\deg s(z)-k)P_{\infty});
$$
if $k \leq \deg s(z)<k+\deg d(z)+1$, then
$$
Hull(GRS_k(\boldsymbol \alpha, \boldsymbol v))=C_{\mathscr L}(D,(d(z))+(k+\deg d(z)-\deg s(z)-1)P_{\infty}).
$$
\end{cor}
\begin{exa}
Let $q=2^3$, $\Bbb F_q^*=\langle \gamma \rangle$, and $h(z)=z^7-1=\prod_{i=0}^6(z-\gamma^i)\in \Bbb F_q[z]$. Then $h'(z)=z^6\in \Bbb F_q[z]$.

(1) Let $s(z)=z^2$ and $t(z)=z^4\in \Bbb F_q[z]$. Then $\gcd(s(z),t(z))=s(z)$. Let $\boldsymbol \alpha=(\alpha_1=\gamma^0,\alpha_2=\gamma^1,\ldots,\alpha_7=\gamma^6)$ and $\boldsymbol v=(v_1,\ldots,v_7)$, where $v_i=(s(\alpha_i))^{-1}=\gamma^{-2(i-1)},i=1,\ldots,7$. By Corollary \ref{th appl} (2), if $1\leq k\leq 2$, then $GRS_k(\boldsymbol \alpha, \boldsymbol v)$ is a self-orthogonal code.
Specifically, $Hull(GRS_1(\boldsymbol \alpha, \boldsymbol v))=GRS_1(\boldsymbol \alpha, \boldsymbol v)$ are both $[7, 1, 7]$ GRS codes over $\Bbb F_q$, i.e., $GRS_1(\boldsymbol \alpha, \boldsymbol v)$ is a self-orthogonal code over $\Bbb F_q$, and they have a generator matrix as follows:
$$
\left(\begin{array}{ccccccc}
1 & \gamma^5 & \gamma^3 & \gamma & \gamma^6 & \gamma^4 & \gamma^2
\end{array}\right).
$$
$Hull(GRS_2(\boldsymbol \alpha, \boldsymbol v))=GRS_2(\boldsymbol \alpha, \boldsymbol v)$ are both $[7, 2, 6]$ GRS codes over $\Bbb F_q$, i.e., $GRS_2(\boldsymbol \alpha, \boldsymbol v)$ is a self-orthogonal code over $\Bbb F_q$, and they have a generator matrix as follows:
$$
\left(\begin{array}{ccccccc}
1 & 0 & \gamma^4 & \gamma^5 & \gamma^5 & 1 & \gamma^4\\
0 & 1 & \gamma & \gamma & \gamma^3 & 1 & \gamma^3
\end{array}\right).
$$

(2) Let $s(z)=z^4$ and $t(z)=z^2\in \Bbb F_q[z]$. Then $\gcd(s(z),t(z))=t(z)$. Let $\boldsymbol \alpha=(\alpha_1=\gamma^0,\alpha_2=\gamma^1,\ldots,\alpha_7=\gamma^6)$ and $\boldsymbol v=(v_1,\ldots,v_7)$, where $v_i=(s(\alpha_i))^{-1}=\gamma^{-4(i-1)},i=1,\ldots,7$. By Corollary \ref{th appl} (2), if $4<k<7$, then $GRS_k(\boldsymbol \alpha, \boldsymbol v)$ is a dual-containing code.
Specifically, $Hull(GRS_5(\boldsymbol \alpha, \boldsymbol v))=GRS_5(\boldsymbol \alpha, \boldsymbol v)^\perp$ are both $[7, 2, 6]$ GRS codes over $\Bbb F_q$, i.e., $GRS_5(\boldsymbol \alpha, \boldsymbol v)^\perp$ is a self-orthogonal code over $\Bbb F_q$, and they have a generator matrix as follows:
$$
\left(\begin{array}{ccccccc}
1 & 0 & \gamma^4 & \gamma^5 & \gamma^5 & 1 & \gamma^4\\
0 & 1 & \gamma & \gamma & \gamma^3 & 1 & \gamma^3
\end{array}\right).
$$
$Hull(GRS_6(\boldsymbol \alpha, \boldsymbol v))=GRS_6(\boldsymbol \alpha, \boldsymbol v)^\perp$ are both $[7, 1, 7]$ GRS codes over $\Bbb F_q$, i.e., $GRS_6(\boldsymbol \alpha, \boldsymbol v)^\perp$ is a self-orthogonal code over $\Bbb F_q$, and they have a generator matrix as follows:
$$
\left(\begin{array}{ccccccc}
1 & \gamma^5 & \gamma^3 & \gamma & \gamma^6 & \gamma^4 & \gamma^2
\end{array}\right).
$$
\end{exa}
\begin{cor}
In (\ref{im eq}), let $\deg u(z)=0$. If $k-\deg d(z)-1<\deg s(z)<k+\deg d(z)+1$, then the hull of $GRS_k(\boldsymbol \alpha, \boldsymbol v)$ is also a GRS code. Moreover,
$$
Hull(GRS_k(\boldsymbol \alpha, \boldsymbol v))=C_{\mathscr L}(D,(d(z))+(\deg d(z)-1-\lvert\deg s(z)-k\rvert)P_{\infty}).
$$
\end{cor}
\begin{exa}
Let $q=5^2$ and $\Bbb F_q^*=\langle \gamma \rangle$.

(1) Let $h(z)=z^4+\gamma^5 z^3+\gamma^{14} z^2+\gamma^{23} z+\gamma^8 =(z-\gamma^8)(z-4)(z-\gamma^{17})(z-\gamma^{19})\in \Bbb F_q[z]$. Then $h(z)+h'(z)=((z+\gamma^2)(z+\gamma^{15}))^2\in \Bbb F_q[z]$ by Magma \cite{m}.
Let $s(z)=(z+\gamma^2)(z+\gamma^{15})\in \Bbb F_q[z]$,
$\boldsymbol \alpha=(\alpha_1,\alpha_2,\alpha_3,\alpha_4)=(\gamma^8,4,\gamma^{17},\gamma^{19})$, and
$\boldsymbol v=(v_1,v_2,v_3,v_4)=(4,\gamma^{11},4,\gamma^7)$, where $\alpha_i$ are the roots of $h(z)$ and $v_i=(s(\alpha_i))^{-1},i=1,\ldots,4$. Then $Hull(GRS_2(\boldsymbol \alpha, \boldsymbol v))=GRS_2(\boldsymbol \alpha, \boldsymbol v)$ are both $[4, 2, 3]$ GRS codes over $\Bbb F_q$, i.e., $GRS_2(\boldsymbol \alpha, \boldsymbol v)$ is a self-dual code over $\Bbb F_q$, and they have a generator matrix as follows:
$$
\left(\begin{array}{cccc}
1 & 0 & \gamma^{21} & 1 \\
0 & 1 & 4 & \gamma^{21}
\end{array}\right).
$$

(2) Let $h(z)=z^4 + \gamma^{23}z^3 + \gamma^{17}z^2 + 3z + \gamma^{23} =(z-\gamma^5)(z-\gamma^9)(z-\gamma^{15})(z-3)\in \Bbb F_q[z]$. Then $h(z)+h'(z)=((z+\gamma^9)(z+\gamma^{19}))^2\in \Bbb F_q[z]$ by Magma \cite{m}.
Let $s(z)=(z+\gamma^{9})(z+\gamma^{19})\in \Bbb F_q[z]$,
$\boldsymbol \alpha=(\alpha_1,\alpha_2,\alpha_3,\alpha_4)=(\gamma^5,\gamma^9,\gamma^{15},3)$, and
$\boldsymbol v=(v_1,v_2,v_3,v_4)=(4,\gamma^9,\gamma^7,4)$, where $\alpha_i$ are the roots of $h(z)$ and $v_i=(s(\alpha_i))^{-1},i=1,\ldots,4$. Then $Hull(GRS_2(\boldsymbol \alpha, \boldsymbol v))=GRS_2(\boldsymbol \alpha, \boldsymbol v)$ are both $[4, 2, 3]$ GRS codes over $\Bbb F_q$, i.e., $GRS_2(\boldsymbol \alpha, \boldsymbol v)$ is a self-dual code over $\Bbb F_q$, and they have a generator matrix as follows:
$$
\left(\begin{array}{cccc}
1 & 0 & \gamma^3 & \gamma^3 \\
0 & 1 & \gamma^{15} & \gamma^3
\end{array}\right).
$$
\end{exa}

\begin{cor}
In (\ref{im eq}), let $\deg u(z)=1$. If $k-\deg d(z)<\deg s(z)<k+\deg d(z)+1$, then the hull of $GRS_k(\boldsymbol \alpha, \boldsymbol v)$ is also a GRS code. Moreover,
if $k-\deg d(z) < \deg s(z)\le k$, then
$$
Hull(GRS_k(\boldsymbol \alpha, \boldsymbol v))=C_{\mathscr L}(D,(d(z))+(\deg d(z)+\deg s(z)-k-2)P_{\infty});
$$
if $k+1 \leq \deg s(z)<k+\deg d(z)+1$, then
$$
Hull(GRS_k(\boldsymbol \alpha, \boldsymbol v))=C_{\mathscr L}(D,(d(z))+(k+\deg d(z)-\deg s(z)-1)P_{\infty}).
$$
\end{cor}

\begin{exa}
Let $q=5^2$ and $\Bbb F_q^*=\langle \gamma \rangle$.

(1) Let $h(z)=z^4+\gamma^3 z^3+\gamma^{22} z^2+\gamma^{17} z+\gamma^{11} =(z-\gamma)(z-\gamma^7)(z-\gamma^{13})(z-\gamma^{14})\in \Bbb F_q[z]$.
Then $zh(z)+h'(z)=(z+\gamma^9)(z^2+\gamma^9z+\gamma^{16})^2\in \Bbb F_q[z]$ by the Magma program.
Let $s(z)=z^2+\gamma^9z+\gamma^{16} \in \Bbb F_q[z]$,
$\boldsymbol \alpha=(\alpha_1,\alpha_2,\alpha_3,\alpha_4)=(\gamma,\gamma^7,\gamma^{13},\gamma^{14})$, and
$\boldsymbol v=(v_1,v_2,v_3,v_4)=(\gamma^5,2,3,\gamma)$, where $\alpha_i$ are the roots of $h(z)$ and $v_i=(s(\alpha_i))^{-1},i=1,\ldots,4$. By Corollary \ref{th appl} (2), if $k=1$, then $GRS_k(\boldsymbol \alpha, \boldsymbol v)$ is a self-orthogonal code.
Specifically, $Hull(GRS_1(\boldsymbol \alpha, \boldsymbol v))=GRS_1(\boldsymbol \alpha, \boldsymbol v)$ are both $[4, 1, 4]$ GRS codes over $\Bbb F_q$, i.e., $GRS_1(\boldsymbol \alpha, \boldsymbol v)$ is a self-orthogonal code over $\Bbb F_q$, and they have a generator matrix as follows:
$$
\left(\begin{array}{cccc}
1 & \gamma & \gamma^{13} & \gamma^{20}
\end{array}\right).
$$

(2) Let $h(z)=z^4 + \gamma^{4}z^3 + 3z^2 + \gamma z + \gamma^3 =(z-\gamma^3)(z-\gamma^9)(z-\gamma^{19})(z-\gamma^{20})\in \Bbb F_q[z]$. Then $zh(z)+h'(z)=(z+\gamma^{17})(z+\gamma^{14})^4\in \Bbb F_q[z]$ by Magma \cite{m}.
Let $s(z)=(z+\gamma^{17})(z+\gamma^{14})^2\in \Bbb F_q[z]$,
$\boldsymbol \alpha=(\alpha_1,\alpha_2,\alpha_3,\alpha_4)=(\gamma^3,\gamma^9,\gamma^{19},\gamma^{20})$, and
$\boldsymbol v=(v_1,v_2,v_3,v_4)=(\gamma^{10},\gamma^2,\gamma^{16},\gamma^8)$, where $\alpha_i$ are the roots of $h(z)$ and $v_i=(s(\alpha_i))^{-1},i=1,\ldots,4$. By Corollary \ref{th appl} (2), if $k=3$, then $GRS_k(\boldsymbol \alpha, \boldsymbol v)$ is a dual-containing code.
Specifically, $Hull(GRS_3(\boldsymbol \alpha, \boldsymbol v))=GRS_3(\boldsymbol \alpha, \boldsymbol v)^\perp$ are both $[4, 1, 4]$ GRS codes over $\Bbb F_q$, i.e., $GRS_3(\boldsymbol \alpha, \boldsymbol v)^\perp$ is a self-orthogonal code over $\Bbb F_q$, and they have a generator matrix as follows:
$$
\left(\begin{array}{cccc}
1 & \gamma^{16} & 2 & \gamma^{22}
\end{array}\right).
$$
\end{exa}

\begin{cor}
In (\ref{im eq}), let $\deg u(z)=2$ and $\deg d(z)\geq 1$. If $k-\deg d(z)+1<\deg s(z)<k+\deg d(z)+1$, then the hull of $GRS_k(\boldsymbol \alpha, \boldsymbol v)$ is also a GRS code. Moreover,
if $k-\deg d(z)+1< \deg s(z)<k+1$, then
$$
Hull(GRS_k(\boldsymbol \alpha, \boldsymbol v))=C_{\mathscr L}(D,(d(z))+(\deg d(z)+\deg s(z)-k-3)P_{\infty});
$$
if $k+1 \leq \deg s(z)<k+\deg d(z)+1$, then
$$
Hull(GRS_k(\boldsymbol \alpha, \boldsymbol v))=C_{\mathscr L}(D,(d(z))+(k+\deg d(z)-\deg s(z)-1)P_{\infty}).
$$
\end{cor}

\begin{exa}
Let $q=7^2$ and $\Bbb F_q^*=\langle \gamma \rangle$. Let $h(z)=z^4 + z =z(z-\gamma^8)(z-\gamma^{24})(z-\gamma^{40})\in \Bbb F_q[z]$. Then $z^2h(z)+h'(z)=((z+\gamma^8)(z+\gamma^{24})(z+\gamma^{40}))^2\in \Bbb F_q[z]$ by Magma \cite{m}.
Let $s(z)=(z+\gamma^{8})(z+\gamma^{24})(z+\gamma^{40})\in \Bbb F_q[z]$,
$\boldsymbol \alpha=(\alpha_1,\alpha_2,\alpha_3,\alpha_4)=(0,\gamma^8,\gamma^{24},\gamma^{40})$, and
$\boldsymbol v=(v_1,v_2,v_3,v_4)=(\gamma^{24},\gamma^8,\gamma^8,\gamma^8)$, where $\alpha_i$ are the roots of $h(z)$ and $v_i=(s(\alpha_i))^{-1},i=1,\ldots,4$. Then $Hull(GRS_2(\boldsymbol \alpha, \boldsymbol v))=GRS_2(\boldsymbol \alpha, \boldsymbol v)$ are both $[4, 2, 3]$ GRS codes over $\Bbb F_q$, i.e., $GRS_2(\boldsymbol \alpha, \boldsymbol v)$ is a self-dual code over $\Bbb F_q$, and they have a generator matrix as follows:
$$
\left(\begin{array}{cccc}
1 & 0 & \gamma^8 & \gamma^{16}\\
0 & 1 & \gamma^{16} & \gamma^{32}
\end{array}\right).
$$
\end{exa}
\section{Conclusion and future work}
This study aims to determine hulls of GRS codes. First of all, a GRS code and its dual code are represented as rational algebraic geometric codes. Then the conditions that hulls of GRS codes are also GRS codes are provided by using the language of algebraic geometry codes. If the conditions are not satisfied, a method for finding the bases of hulls of GRS codes and formulas to compute the hull dimension are given. Besides, some examples are provided to show the conditions are too good to be improved. Moreover, some self-orthogonal GRS codes and self-dual GRS codes are shown, which  generalize results in \cite{GY}.

 An interesting question to be considered in the future is the search for Hermitian hulls of GRS codes.

\end{document}